 \documentclass{sig-alternate-m}

\usepackage{latexsym}

\usepackage{amssymb,amsmath,amsfonts,bm}
\usepackage{mathtools}
\usepackage{subfigure}
\usepackage{epic}
\usepackage{eepic}
\usepackage{enumerate}

\usepackage{amsthm}
\usepackage{vinayak}
\usepackage{graphicx, color}
\usepackage{verbatim}
\usepackage[mathscr]{eucal}
\usepackage{xspace}
\usepackage{stmaryrd}
\usepackage{paralist}
\makeatletter
\newif\if@restonecol
\makeatother

\usepackage[boxruled]{algorithm2e}
 \SetAlgoInsideSkip{smallskip}
\SetAlCapSkip{9pt}
\SetAlFnt{\small\rm}

\renewcommand{\baselinestretch}{0.99}

\makeatletter
\let\oldabs\abs
\def\abs{\@ifstar{\oldabs}{\oldabs*}}
\let\oldnorm\norm
\def\norm{\@ifstar{\oldnorm}{\oldnorm*}}
\makeatother

\makeatletter
\g@addto@macro \normalsize {%
 \setlength\abovedisplayskip{5pt plus 2pt minus 2pt}%
 \setlength\belowdisplayskip{4pt plus 2pt minus 2pt}%
}
\makeatother

\begin{document}

\newcommand{\RM}[1]{{\textcolor{red}{ \textbf{RM:} #1 }}}
\newcommand{\VP}[1]{{\textcolor{blue}{ \textbf{VP:} #1 }}}

\title{
Computing Distances between Reach Flowpipes
}
\author{
Rupak Majumdar
\qquad\qquad
Vinayak S. Prabhu
\normalsize
}

\sloppy

\maketitle

\begin{abstract}
We investigate quantifying the difference between two hybrid dynamical systems
under noise and initial-state uncertainty.
While the set of traces  for these systems is infinite, it is possible to symbolically 
approximate trace sets using \emph{reachpipes} that compute
upper and lower bounds on the evolution of the reachable sets with time.
We estimate distances between corresponding sets of trajectories of two systems
in terms of distances between the reachpipes.

In case of two individual traces, the Skorokhod distance has been proposed as a robust
and efficient notion of distance which captures both value and timing distortions.
In this paper, we extend the computation of the Skorokhod distance to reachpipes, and
provide algorithms to compute upper and lower bounds on the 
distance between two sets of traces.
Our algorithms use new geometric insights that are used to 
compute the worst-case and best-case distances
between two polyhedral sets evolving with time.

\end{abstract}



\section{Introduction}

The quantitative conformance problem between two dynamical systems asks 
how close the traces of the two systems are under a given metric on hybrid
traces \cite{GeorgiosHFDKU14,AbbasMF14,DMP15}.
If the systems are deterministic and start from unique initial conditions, each
has exactly one trace, and the quantitative conformance problem computes the distance
between these two traces.
In this case,  we have shown in previous work
 that the \emph{Skorokhod metric} between traces
provides a robust and efficiently computable distance that captures the intuitive
notion of closeness of two systems \cite{MajumdarP15,DMP15}. 
However, if there is  uncertainty in the initial states and noise in the inputs,
each system defines not just a single trace but a set of traces.
In this work, we investigate algorithms to compute 
distances between  sets of trajectories of 
two dynamical systems under initial state and
input uncertainties.

Given two sets $F_1, F_2$ of trajectories of two dynamical systems, 
the natural generalization of the Skorokhod distance between traces
is to ask what is the farthest  a trajectory in one set can be from
a trajectory in the other, \ie, to compute 
\[
\dvar(F_1, F_2) = \sup_{f_1, f_2} \dist_{\tr}(f_1, f_2)
\]
where $\dist_{\tr}$
is the  given Skorokhod metric on traces\footnote{
In comparing sets, we use the term ``distance'' for similarity/dissimilarity functions $\dvar$
satisfying the triangle inequality; these functions 
are not necessarily metrics, as  $\dvar(F,F)$ need  not be  zero.
}.

Unfortunately, due to the continuous nature of systems, trace sets $F_1$ and $F_2$
are not available in closed form for most kinds of  systems.
Instead, given a trace set $F$,
one approximates it  using a \emph{reachpipe},
a function $R: [0, T]\rightarrow 2^{\reals^d}$, such that 
$R(t) = \cup_{f\in F}\set{f(t)}$, \ie, $R(t)$ is the set of all trace values that can be observed
at time $t$. 
A  reachpipe $R$ can be viewed as an approximation $\fpipe(R)$
to the original set of traces, the approximation $\fpipe(R)$ includes every
 trace $f$ such that $f(t) \in R(t)$,
 not just those allowed by the dynamics.
In practice, even the reachpipe may not have an exact representation, 
and instead, one computes over- or under-approximations to the reachpipe by
computing a  sequence of \emph{reach set} samples at discrete timepoints $t_0, t_1, \dots$.
Indeed, there are several techniques to compute such approximations of
reach sets 
\cite{ChutinanK03,KurzhanskiV06,Girard05,GirardGM06,GuernicG10,FrehseGDCRLRGDM11,SankaranarayananDI08,ColonS11,ChenAS12},
differing in the quality of the approximation, the efficiency of computation, or the representation of the reach set approximations.

We consider the problem of estimating  trajectory set distances when we only have
the sampled sequences of over- and under-approximations of reach sets. 
As a first step, we define a lower and an upper bound on the distance between $F_1$ and $F_2$ based on the
reach set approximations.

Second, we show how to compute these bounds. 
To compute the distance, we re-formulate reachpipes as set-valued traces, \ie, as traces over the
time interval $[0,T]$ where the trace value at time $t$ is the set $R(t)\subseteq \reals^d$.
This alternative viewpoint allows us to define  trace distances $\dist^{\dagger}$
between reachpipes by viewing them as set-valued traces.
We derive relationships between the distances 
$\dist^{\dagger}$ under this alternative viewpoint, 
and distances bounding the trace set distance (obtained using 
approximations to the reachpipes).

Finally, we derive  algorithms to compute the 
$\dist^{\dagger}$ distances between reachpipes in case the underlying metric on traces
is given by the Skorokhod distance and 
the reach set sequences are given as polytopes in $\reals^d$.
The Skorokhod distance on traces takes into account both timing distortions and value differences;
our algorithms lift the metric to reach sets viewed as time-varying polytopes.
The algorithms allows for timing distortions,  
and generalize the Skorokhod distance algorithm
over polygonal lines to  polytopes which vary with time. 
The main technical constructions in  our algorithms  are two novel geometric routines
in  a core part of the  Skorokhod distance algorithm which allow us to move to
the domain of time-varying polytopes for the set distances under consideration.

Putting everything together, we obtain  polynomial time
 algorithms which compute
bounds on traceset distances where the tracesets are 
observed only as reachset sample-polytopes at discrete timepoints.

\smallskip\noindent\textbf{Outline of the Paper.}
In Section~\ref{section:Metrics}, we recall  the Skorokhod trace metric, and the related
\frechet metric.
In Section~\ref{section:Pipes}, we formally present tracepipes and reachpipes,  distances
between trace sets, and bounds on these set distances.
In Section~\ref{section:FrechetPipe} we explore the alternative viewpoint of reachpipes 
being 
set valued traces, and relate  distances under this viewpoint and
distances between reachpipes viewed as trace sets.
In Section~\ref{section:ComputeFrechet}, we solve for the  distance
decision problems  between
reachpipes viewed as time-varying polytopes of $\reals^d$.
In Section~\ref{section:FinalAlgo} we put everything together and present 
various 
algorithms to compute bounds on  Skorokhod traceset distances.

\section{Preliminaries: Trace Metrics}
\label{section:Metrics}

A (finite) \emph{trace} 
$f: [T_i,T_e] \rightarrow  \reals^d $ is a  continuous mapping from a
finite closed interval $[T_i,T_e]$ of $\reals_+$,
with $0 \leq T_i < T_e$, to  $\reals^d$.

\subsection{The Skorokhod Trace Metric}

We define a metric on the space of traces corresponding to a given metric on $\reals^d$. 
A \emph{retiming} $\retime: I \mapsto  I' $, for closed
intervals $I, I'$ of $\reals_+$, 
is an order-preserving (i.e., monotone) continuous bijective function from
$I$ to $I'$; 
thus if $t<t'$ then $\retime(t) < \retime(t')$.
Let $\retimeclass_{I \mapsto  I' }$
be the class of retiming functions from $I$ to $ I' $
and let $\iden$ be the identity retiming.
Given a trace  $f: I_{f} \rightarrow \reals^d$, and a retiming
$\retime: I \mapsto   I_{f} $; the function
$f\circ \retime$ is another trace from $I$ to $\reals^d$.

\begin{definition}[Skorokhod Metric]
Given a retiming $\retime:  I \mapsto I' $, define 
\[
||\retime-\iden||_{\sup} := \sup_{t\in I }|\retime(t)-t|.\]
Given two traces $f: I_{f}\mapsto \reals^d $ and $f': I_{f'} \mapsto \reals^d$,
a norm  $L$ on $\reals^d$,
and a retiming $\retime:  I_{f} \mapsto   I_{f'}$, define
\[
\norm{f\,-\, f'\circ \retime}_{\sup}
:=
\sup\nolimits_{t\in I_{f}} \norm{ f(t) -  f'\left(\retime(t)\right)}_L.
\]
The \emph{Skorokhod metric}\footnote{
	The two components of the Skorokhod metric (the retiming, and the value difference components) can
	be weighed with different weights -- this simply corresponds to a change of scale.}
between the traces  $f$ and $f'$  is defined to be:
\begin{equation*}
\label{equation:Skoro}
\dist_{\skoro}(f,f') := \inf_{r\in  \retimeclass_{ I_{f} \mapsto   I_{f'}}}
\max\left(\norm{\retime-\iden}_{\sup} \, ,\,  \norm{f\,-\, f'\circ \retime}_{\sup}\right).\qed
\end{equation*}
\end{definition}

Intuitively, the  Skorokhod metric
incorporates two components: the first component quantifies
the {\em timing discrepancy} of the timing distortion required to ``match'' the two traces,
and the second quantifies the  \emph{value mismatch}  (in the vector space 
$(\reals^d, \norm{\cdot}_L)$)
of the values under the timing distortion.
In the retimed trace $f\circ \retime$, we see exactly the same values as in $f$, in
exactly the same order, but the times at which the values are seen can be different.


\subsection{The \frechet Trace Metric}
\label{subsection:MovingFrechet}

We showed in~\cite{MajumdarP15} that the Skorokhod metric is related to another metric,
the \frechet metric, over traces.
We recall the definition and the relationship.

\begin{definition}[\frechet metric]
\label{def:frechet}
Let $\curve_1:I_1 \rightarrow \reals^d$ and $\curve_2:I_2 \rightarrow \reals^d$ be traces.
The \frechet metric between the two traces $\curve_1,\curve_2$ (given a norm $L$ on
$\reals^d$) is defined to be
\[
\mspace{-5mu}
\dist_{\fre}(\curve_1,\curve_2) :=\! \inf_{\substack{
\alpha_1: [0,1] \rightarrow I_1 \\
\alpha_2: [0,1] \rightarrow I_2 }}\quad
\max_{0\leq \theta \leq 1} \norm{\curve_1\left(\alpha_1(\theta)\right) - \curve_2\left(\alpha_2(\theta)\right)}_L
\]
where $\alpha_1, \alpha_2$ range over  continuous and strictly increasing bijective functions
onto $I_1$ and  $I_2$, respectively.\qed
\end{definition}
Intuitively, the \emph{reparameterizations} $\alpha_1, \alpha_2$ control the ``speed'' of traversal along the
two traces by two entities.
The positions of the two entities in the two traces at ``time'' $\theta$  is given by
$\alpha_1(\theta)$ and $\alpha_2(\theta)$ respectively; with the value of the traces
at those positions being
$\curve_1\left(\alpha_1(\theta)\right) $, and  $\curve_2\left(\alpha_2(\theta)\right) $.
The two entities always have a speed strictly greater than $0$.

Given a trace $f: [T_i,T_e] \rightarrow  \reals^d $, we define 
the \emph{time-explicit trace} $\curve_f: [T_i,T_e] \rightarrow  \reals^{d}\times \reals $ where we 
add the time value as an  extra dimension, that is, $\curve_f(t) = (f(t), t)$ for all $t\in [T_i,T_e]$.
Given a value $\tuple{\bp, t} \in \reals^{d}\times \reals$, and a 
a norm $L$
 over
$\reals^d$, define the norm 
\begin{equation}
\label{equation:LMaxDef}
\norm{\tuple{\bp, t}}_{L^{\max}} = \max\left( \norm{\bp}_L, \abs{t}\right).
\end{equation}

\begin{proposition}[From Skorokhod to \frechet~\cite{MajumdarP15}]
\label{proposition:SkoroToFrechet}
Let $f: [T_i^f, T_e^f] \rightarrow \reals^d$ and $g: [T_i^g, T_e^g] \rightarrow \reals^d$ be
two  continuous traces.
Consider the corresponding time-explicit traces  $C_f : [T_i^f, T_e^f]\rightarrow \reals^{d+1}$ and
$C_g : [T_i^g, T_e^g]\rightarrow \reals^{d+1}$.
Consider the Skorokhod distance $\dist_{\skoro}(f,g)$  with respect to  a given norm $L$ over 
$\reals^d$.
We have
\[
\dist_{\skoro}(f,g) =\dist_{\fre}(\curve_f, \curve_g),
 \]
where the \frechet distance $\dist_{\fre}(\curve_f, \curve_g)$ is with respect to
the norm $L^{\max}$ over $ \reals^{d+1}$.\qed
\end{proposition}

\section{Pipes \& Pipe-Variation Distances}
\label{section:Pipes}

\subsection{Tracepipes, Reachpipes and Set Distances}

A \emph{tracepipe} $F$  is a nonempty collection of traces over some closed interval $[T_i,T_e]$.
A \emph{reachpipe} $R: [T_i, T_e] \rightarrow 2^{\reals^d}\setminus \emptyset$ maps
a finite closed interval $[T_i,T_e]$  of $\reals_+$, denoted $\tdom(R)$, 
to non-empty subsets of $\reals^d$.
To a reachpipe $R$, we associate a tracepipe $\fpipe(R)$ consisting of all continuous traces
$f$  over $\tdom(R) $ such that $f(t) \in R(t)$ for all $t\in \tdom(R)$.
Dually, corresponding to each tracepipe $F$, we associate the reachpipe $\rpipe(F)$, 
over the same time-domain, defined by 
$\rpipe\!\left(F\right)\left(t\right) =\cup_{f\in F} \set{f(t)}$.
Note that $F \subseteq \fpipe\left(\rpipe\left(F\right)\right)$,
but equality need not hold:
$ \fpipe\left(\rpipe\left(F\right)\right)$ may contain more traces than $F$.

A reachpipe $R': [T_i,T_e] \rightarrow 2^{\reals^d}$ is an \emph{over-approximation}
(respectively, \emph{under-approximation})
of a reachpipe $R: [T_i,T_e] \rightarrow 2^{\reals^d}$ if for each $t\in [T_i, T_e]$, we have 
$R(t) \subseteq R'(t)$
(respectively, $R'(t) \subseteq R(t)$).

\begin{example}
Consider a linear dynamical system in $\reals$ described by
$\dot{x} = a x$, for $a > 0$ with initial state $x_0\in [0, 0.1]$
over the time interval $[0,10]$.
For a fixed value of $x_0$, we get a trace $x_0 e^{at}$.
Let $F = \set{ f_{x_0} \mid x_0\in [0, 0.1] \mbox{ and } f_{x_0}(t) = x_0 e^{at}\mbox{ for }t\in[0,10]}$ be a tracepipe.
The reachpipe $\rpipe(F)$ corresponding to the tracepipe $F$ is given by
$\rpipe(F) (t) = [0, 0.1e^{at}]$ for $t\in [0,10]$.
Observe that $\fpipe\left(\rpipe(F) \right)$ contains the more traces than 
$F$, for instance, the constant trace $f(t)= 0.1$.\qed
\end{example}

Let $\dist_{\tr}$ be a given metric on traces.
We define the \emph{variation distance}
$\dvar(F_1, F_2)$ 
between
two tracepipes $F_1$ and $F_2$ corresponding to the trace metric  $\dist_{\tr}$  as 
\begin{equation}
\dvar(F_1, F_2) := \sup_{f_1\in F_1, f_2\in F_2} \dist_{\tr}(f_1, f_2)
\end{equation}
The value $\dvar(F_1, F_2)$ gives us the maximum possible inter-trace distance
if one trace is  from $F_1$ and the other from $F_2$.
Notice that for all tracepipes $F_1, F_2, F_3$, we have that
\begin{compactenum}
\item 
$\dvar(F_1, F_2) \geq 0$;
\item $\dvar(F_1, F_2) = \dvar(F_2, F_1)$; and
\item 
$\dvar(F_1, F_3) \leq \dvar(F_1, F_2) + \dvar(F_2, F_3) $.
\end{compactenum}
We may however have $\dvar(F, F) > 0$, thus, $\dvar$ need not be a metric over
tracepipes.
The value $\dvar(F, F) $ gives us the maximum
distance amongst traces in $F$ according to the original trace metric $\dist_{\tr}$.




Tracepipes cannot be constructed for most dynamical systems.
However, \emph{reachpipe} sets can be over/under-approximated at desired 
timepoints using analytic techniques.
In the next subsection, we present a  framework for bounding the tracepipe
distance $\dvar(F_1, F_2) $ using over/under-approximated reachpipes.

\subsection{Approximating the Variation Distance}
Let $F_1$ and $F_2$ be tracepipes. 
Since $F\subseteq \fpipe(\rpipe(F))$ for any tracepipe $F$, and $\rpipe$, $\fpipe$, and the variation distance $\dvar$
are all monotonic, we have that
\begin{equation}
\label{equation:OverApproxReach}
\dvar(F_1, F_2) \leq \dvar\big(\fpipe\left(\ceil{\rpipe\left(F_1\right)}\right), 
\fpipe\left(\ceil{\rpipe\left(F_2\right)}\right)
 \big)
\end{equation} 
for any over-approximations
$\ceil{\rpipe(F_1)}$ and $\ceil{\rpipe(F_2)}$ of the reachpipes
$\rpipe(F_1)$ and $\rpipe(F_2)$.
Thus, in order to get  an upper bound on  $\dvar(F_1, F_2)$  we can use over-approximations
of the corresponding reachpipes.

Define the \emph{minimum set distance}:
\begin{equation}
\dist_{\min}(F_1, F_2) := \inf_{f_1\in F_1, f_2\in F_2} \dist(f_1, f_2)
\end{equation}
For this distance, it is clear that
\[
 \dist_{\min}\big( \fpipe\left( \rpipe(F_1)\right),  \fpipe\left( \rpipe(F_2)\right)\big)
\leq   \dvar(F_1, F_2) 
\]
Combining this with Equation~\eqref{equation:OverApproxReach}, we get the following
Proposition for bounding the variation distance.

\begin{proposition}[Tracepipe Variation Distance Bounds]
\label{proposition:OverUnder}
Let $F_1$ and $F_2$ be tracepipes, and let
$\ceil{\rpipe(F_1)}$ and $\ceil{\rpipe(F_2)}$ be over-approximations of the 
reachpipes $\rpipe(F_1)$ and $\rpipe(F_2)$.
We have
\begin{gather*}
\dist_{\min}\!\Big( \fpipe\left( \ceil{\rpipe(F_1)}\right),  \fpipe\left( \ceil{\rpipe(F_2)}\right)\!\Big) 
\ \leq\    \dvar(F_1, F_2)\\
 \dvar(F_1, F_2)
\ \leq\  \dvar\!\Big( \fpipe\left( \ceil{\rpipe(F_1)}\right),  \fpipe\left( \ceil{\rpipe(F_2)}\right)\!\Big) 
\qedhere\qed
\end{gather*}
\end{proposition}

\smallskip\noindent\textbf{Remark: Hausdorff Metric.}
A natural candidate for under-approximating the variation distance is the
\emph{Hausdorff} set metric, defined as:
\begin{equation}
\mspace{-8mu}\dist_H(F_1, F_2) = \max\!\left\{
\sup_{f_1 \in F_1}\inf_{f_2\in F_2}\dist(f_1, f_2)\, , \, \sup_{f_2 \in F_2}\inf_{f_1\in F_1}
\dist(f_1, f_2)\!\right\}
\end{equation}
Intuitively, if  $\sup_{f_1 \in F_1}\inf_{f_2\in F_2}\dist(f_1, f_2)$ is less than $\delta$,
then given any trace $f_1\in F_1$, there exists a trace $f_2\in F_2$ such that
$ \dist(f_1, f_2) < \delta$.
Note that $\sup_{f_1 \in F_1}\inf_{f_2\in F_2}\dist(f_1, f_2) \leq \dvar(F_1, F_2)$ and also
 $\sup_{f_2 \in F_2}\inf_{f_1\in F_1}\dist(f_1, f_2) \leq \dvar(F_1, F_2)$, thus, we have
\begin{equation}
\dist_H(F_1, F_2)  \leq \dvar(F_1, F_2)
\end{equation}
Thus, on first glance,  the Hausdorff metric appears to be a good candidate for under-approximating the variation distance.
As mentioned earlier, obtaining tracepipe sets is usually not possible; we have to work with
over or under-approximations obtained by way of reachpipes.
Unfortunately, there is no obvious relationship between
$\dist_H(A, B) $ and $\dist_H(A', B')$ for $A\subseteq A'$ and $B\subseteq B'$.
This can be seen pictorially in Figure~\ref{figure:Hausdorff}.
The sets $A, B, A', B'$ are subsets of the interval $[0,10]$.
In the first case, we have $\dist_H(A, B)  > \dist_H(A', B')$ and in the second,
 $\dist_H(A, B)  < \dist_H(A', B')$.

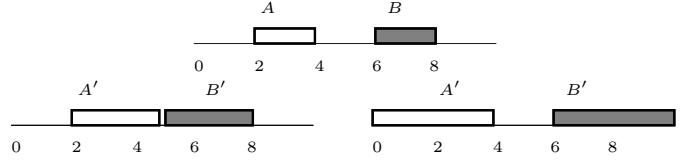
\begin{figure}[h]
\vspace*{-0.5em}
\strut\centerline{\setlength{\unitlength}{0.00034996in}
\begingroup\makeatletter\ifx\SetFigFont\undefined%
\gdef\SetFigFont#1#2#3#4#5{%
  \reset@font\fontsize{#1}{#2pt}%
  \fontfamily{#3}\fontseries{#4}\fontshape{#5}%
  \selectfont}%
\fi\endgroup%
{\renewcommand{\dashlinestretch}{30}
\begin{picture}(4527,1056)(0,-10)
\path(15,420)(4515,420)
\path(15,420)(4515,420)
\thicklines
\texture{44555555 55aaaaaa aa555555 55aaaaaa aa555555 55aaaaaa aa555555 55aaaaaa 
	aa555555 55aaaaaa aa555555 55aaaaaa aa555555 55aaaaaa aa555555 55aaaaaa 
	aa555555 55aaaaaa aa555555 55aaaaaa aa555555 55aaaaaa aa555555 55aaaaaa 
	aa555555 55aaaaaa aa555555 55aaaaaa aa555555 55aaaaaa aa555555 55aaaaaa }
\shade\path(2715,645)(3615,645)(3615,420)
	(2715,420)(2715,645)
\path(2715,645)(3615,645)(3615,420)
	(2715,420)(2715,645)
\whiten\path(915,645)(1815,645)(1815,420)
	(915,420)(915,645)
\path(915,645)(1815,645)(1815,420)
	(915,420)(915,645)
\put(15,15){\makebox(0,0)[lb]{\smash{{\SetFigFont{6}{7.2}{\familydefault}{\mddefault}{\updefault}0}}}}
\put(915,15){\makebox(0,0)[lb]{\smash{{\SetFigFont{6}{7.2}{\familydefault}{\mddefault}{\updefault}2}}}}
\put(1815,15){\makebox(0,0)[lb]{\smash{{\SetFigFont{6}{7.2}{\familydefault}{\mddefault}{\updefault}4}}}}
\put(2670,15){\makebox(0,0)[lb]{\smash{{\SetFigFont{6}{7.2}{\familydefault}{\mddefault}{\updefault}6}}}}
\put(3525,15){\makebox(0,0)[lb]{\smash{{\SetFigFont{6}{7.2}{\familydefault}{\mddefault}{\updefault}8}}}}
\put(1005,870){\makebox(0,0)[lb]{\smash{{\SetFigFont{6}{7.2}{\familydefault}{\mddefault}{\updefault}$A$}}}}
\put(2895,870){\makebox(0,0)[lb]{\smash{{\SetFigFont{6}{7.2}{\familydefault}{\mddefault}{\updefault}$B$}}}}
\end{picture}
}}
\hspace*{-3mm}
\begin{minipage}[t]{0.22\textwidth}
\setlength{\unitlength}{0.00034996in}
\begingroup\makeatletter\ifx\SetFigFont\undefined%
\gdef\SetFigFont#1#2#3#4#5{%
  \reset@font\fontsize{#1}{#2pt}%
  \fontfamily{#3}\fontseries{#4}\fontshape{#5}%
  \selectfont}%
\fi\endgroup%
{\renewcommand{\dashlinestretch}{30}
\begin{picture}(4527,1056)(0,-10)
\path(15,420)(4515,420)
\path(15,420)(4515,420)
\thicklines
\whiten\path(915,645)(2220,645)(2220,420)
	(915,420)(915,645)
\path(915,645)(2220,645)(2220,420)
	(915,420)(915,645)
\texture{44555555 55aaaaaa aa555555 55aaaaaa aa555555 55aaaaaa aa555555 55aaaaaa 
	aa555555 55aaaaaa aa555555 55aaaaaa aa555555 55aaaaaa aa555555 55aaaaaa 
	aa555555 55aaaaaa aa555555 55aaaaaa aa555555 55aaaaaa aa555555 55aaaaaa 
	aa555555 55aaaaaa aa555555 55aaaaaa aa555555 55aaaaaa aa555555 55aaaaaa }
\shade\path(2310,645)(3615,645)(3615,420)
	(2310,420)(2310,645)
\path(2310,645)(3615,645)(3615,420)
	(2310,420)(2310,645)
\put(15,15){\makebox(0,0)[lb]{\smash{{\SetFigFont{6}{7.2}{\familydefault}{\mddefault}{\updefault}0}}}}
\put(915,15){\makebox(0,0)[lb]{\smash{{\SetFigFont{6}{7.2}{\familydefault}{\mddefault}{\updefault}2}}}}
\put(1815,15){\makebox(0,0)[lb]{\smash{{\SetFigFont{6}{7.2}{\familydefault}{\mddefault}{\updefault}4}}}}
\put(2670,15){\makebox(0,0)[lb]{\smash{{\SetFigFont{6}{7.2}{\familydefault}{\mddefault}{\updefault}6}}}}
\put(3525,15){\makebox(0,0)[lb]{\smash{{\SetFigFont{6}{7.2}{\familydefault}{\mddefault}{\updefault}8}}}}
\put(1005,870){\makebox(0,0)[lb]{\smash{{\SetFigFont{6}{7.2}{\familydefault}{\mddefault}{\updefault}$A'$}}}}
\put(2895,870){\makebox(0,0)[lb]{\smash{{\SetFigFont{6}{7.2}{\familydefault}{\mddefault}{\updefault}$B'$}}}}
\end{picture}
}
\end{minipage}
$\qquad\ $
\begin{minipage}[t]{0.22\textwidth}
\setlength{\unitlength}{0.00034996in}
\begingroup\makeatletter\ifx\SetFigFont\undefined%
\gdef\SetFigFont#1#2#3#4#5{%
  \reset@font\fontsize{#1}{#2pt}%
  \fontfamily{#3}\fontseries{#4}\fontshape{#5}%
  \selectfont}%
\fi\endgroup%
{\renewcommand{\dashlinestretch}{30}
\begin{picture}(4608,1056)(0,-10)
\path(54,420)(4554,420)
\path(54,420)(4554,420)
\thicklines
\whiten\path(54,645)(1854,645)(1854,420)
	(54,420)(54,645)
\path(54,645)(1854,645)(1854,420)
	(54,420)(54,645)
\texture{44555555 55aaaaaa aa555555 55aaaaaa aa555555 55aaaaaa aa555555 55aaaaaa 
	aa555555 55aaaaaa aa555555 55aaaaaa aa555555 55aaaaaa aa555555 55aaaaaa 
	aa555555 55aaaaaa aa555555 55aaaaaa aa555555 55aaaaaa aa555555 55aaaaaa 
	aa555555 55aaaaaa aa555555 55aaaaaa aa555555 55aaaaaa aa555555 55aaaaaa }
\shade\path(2754,645)(4554,645)(4554,420)
	(2754,420)(2754,645)
\path(2754,645)(4554,645)(4554,420)
	(2754,420)(2754,645)
\put(54,15){\makebox(0,0)[lb]{\smash{{\SetFigFont{6}{7.2}{\familydefault}{\mddefault}{\updefault}0}}}}
\put(954,15){\makebox(0,0)[lb]{\smash{{\SetFigFont{6}{7.2}{\familydefault}{\mddefault}{\updefault}2}}}}
\put(1854,15){\makebox(0,0)[lb]{\smash{{\SetFigFont{6}{7.2}{\familydefault}{\mddefault}{\updefault}4}}}}
\put(2709,15){\makebox(0,0)[lb]{\smash{{\SetFigFont{6}{7.2}{\familydefault}{\mddefault}{\updefault}6}}}}
\put(3564,15){\makebox(0,0)[lb]{\smash{{\SetFigFont{6}{7.2}{\familydefault}{\mddefault}{\updefault}8}}}}
\put(1044,870){\makebox(0,0)[lb]{\smash{{\SetFigFont{6}{7.2}{\familydefault}{\mddefault}{\updefault}$A'$}}}}
\put(2934,870){\makebox(0,0)[lb]{\smash{{\SetFigFont{6}{7.2}{\familydefault}{\mddefault}{\updefault}$B'$}}}}
\end{picture}
}
\end{minipage}
\vspace*{-1em}
\caption{Sets $A, B$, and two cases of $A\subseteq A'$, $B\subseteq B'$}
\label{figure:Hausdorff}
\end{figure}
Thus, we cannot use the reachpipe over-approximations  
$\ceil{\rpipe(F_1)}$ and $\ceil{\rpipe(F_2)}$
to get a lower (or upper) bound on $\dist_H(F_1, F_2)$.
This problem occurs even  even in the case of  exact reachpipes
$\rpipe(F_1), \rpipe(F_2)$ as we may have $F_1\subsetneq \fpipe\left( \rpipe(F_1)\right)$ and 
$F_2\subsetneq \fpipe\left( \rpipe(F_2)\right)$

For the special case where $F_1= \set{f_1}$ is a singleton set, we have
\begin{equation}
\dist_H(\set{f_1}, F_2) = \dvar(\set{f_1}, F_2)
\end{equation}
Thus, in case of a singleton $F_1=\set{f_1}$, the value
$\dist_H\big( \fpipe\left( \rpipe(F_1)\right),  \fpipe\left( \rpipe(F_2)\right)\big)$ is equal to the
RHS of Equation~\eqref{equation:OverApproxReach}, and hence
only gives an upper bound on $\dvar(F_1, F_2)$.

We note that even if we under-approximate the reach sets to obtain
$\fpipe\left(\floor{\rpipe(F_1)}\right)$, and $\fpipe\left(\floor{\rpipe(F_2)}\right)$,
we still do not have a lower bound for the
Hausdorff distance as we cannot tell in which direction the distance changes on taking
subsets (Figure~\ref{figure:Hausdorff}).
In addition,  we may have $F \subsetneq \fpipe\left(\floor{\rpipe(F)}\right)$ as
for a traceset $F$, as $\fpipe\left(\rpipe(F)\right)$ over-approximate $F$,
and competes with the fact that
$\floor{\rpipe(F)}$ under-approximates $\rpipe(F)$.

\subsection{Constructing Reachpipes}
\label{subsection:ReachpipeConstruction}

For most dynamical systems, one cannot get a closed-form representation for 
the set of all traces.
However, reachpipe sets can be over/under-approximated at desired 
timepoints using analytic techniques 
\cite{ChutinanK03,KurzhanskiV06,Girard05,GirardGM06,GuernicG10,FrehseGDCRLRGDM11,SankaranarayananDI08,ColonS11,ChenAS12}.
The  procedure for bounding the tracepipe variation distance in this
paper   operates on reachpipes (the bounding quantities are as in 
Proposition~\ref{proposition:OverUnder}).
As a result it is necessary to choose an appropriate representation
of reachpipes so that the 
distance computation procedure remains tractable.

\smallskip\noindent\textbf{Reachpipe Completion.}
Typically, reachset computation tools give us reach sets at sampled time-points, 
\ie, the tools give us  reachpipe samples $R(t_0), \dots, R(t_m)$
at discrete time-points $t_0, \dots, t_m$.
We need to ``complete'' the reachpipes for intermediate time values.
We do this completion by generalizing linear interpolation using
scaling and Minkowski sums.
%
%
Specifically, we define an over-approximated completion of $R$ in between $t_k, t_{k+1}$ as follows
for $t_k \leq t \leq t_{k+1}$:
\[
\mspace{-10mu}
\ceil{R}(t) = \left\{ 
\bp + \frac{t-t_k}{t_{k+1}-t_k}
\vdot (\bq- \bp)\  \Big\arrowvert\ 
\bp\in  R(t_k) \text{ and } \bq\in  R(t_{k+1})  \right\}.
\]
For a set $A\subseteq \reals^d$,
given $\lambda \in \reals$,  let $\lambda\cdot A$ denote
$\set{\lambda\cdot \bp \mid \bp\in A}$.
The Minkowski sum of two sets $A, B$ is defined as
$A+B = \set{\bp+\bq \mid \bp \in A \text{ and } \bq \in B}$.
We also denote $-1 \cdot A$ by $-A$.
Under this notation, we have
\begin{equation}
\label{equation:CompletionOver}
\ceil{R}(t) = R(t_k) + \frac{t-t_k}{t_{k+1}-t_k}\cdot\left(R(t_{k+1})  - R(t_{k})\right).
\end{equation}

Alternately, one can observe individual traces of the system at discrete times and
complete the trace by linear interpolation at intermediate points.
That is, 
suppose we observe a trace $f$ at discrete points $t_k$ and $t_{k+1}$: $f(t_k) = \bp$ and
$f(t_{k+1}) = \bp'$ and complete the trace as $f(t) = \bp + \frac{t - t_k}{t_{k+1} - t_k}(\bp' - \bp)$
for all points $t_k \leq t \leq t_{k+1}$.
We explain why Equation~\eqref{equation:CompletionOver} is an over-approximation for linearly 
interpolated completions of observed trace samples.
Recall that 
\[R(t)  = \set{\bp \mid \text{ there exists some trace } f \text{ such that }
f(t) = \bp}.\]
Under linear interpolation completion of traces, this set is
\begin{equation}
\label{equation:CompletionExact}
R(t)  = \left\{\bp + \frac{t-t_k}{t_{k+1}-t_k}
\vdot (\bq- \bp)\  \bigg\arrowvert
\begin{array}{l}
\text{ there exists a trace } f\\
 \text{ such that }
f(t_k) = \bp \text{ and }\\
 \ f(t_{k+1}) = \bq
\end{array}
\mspace{-15mu}
\right\}
\end{equation}

In general $R(t) $ as defined in
Equation~\eqref{equation:CompletionExact}
 can be a strict subset of $\ceil{R}(t)$ as defined in 
Equation~\eqref{equation:CompletionOver}.
\begin{figure}[h]
\begin{minipage}[t]{0.2\textwidth}
\setlength{\unitlength}{0.00034996in}
\begingroup\makeatletter\ifx\SetFigFont\undefined%
\gdef\SetFigFont#1#2#3#4#5{%
  \reset@font\fontsize{#1}{#2pt}%
  \fontfamily{#3}\fontseries{#4}\fontshape{#5}%
  \selectfont}%
\fi\endgroup%
{\renewcommand{\dashlinestretch}{30}
\begin{picture}(3354,1819)(0,-10)
\texture{44555555 55aaaaaa aa555555 55aaaaaa aa555555 55aaaaaa aa555555 55aaaaaa 
	aa555555 55aaaaaa aa555555 55aaaaaa aa555555 55aaaaaa aa555555 55aaaaaa 
	aa555555 55aaaaaa aa555555 55aaaaaa aa555555 55aaaaaa aa555555 55aaaaaa 
	aa555555 55aaaaaa aa555555 55aaaaaa aa555555 55aaaaaa aa555555 55aaaaaa }
\shade\path(99,687)(3249,687)(3249,12)
	(99,462)(99,687)
\path(99,687)(3249,687)(3249,12)
	(99,462)(99,687)
\shade\path(102,1117)(3252,1117)(3252,1792)
	(192,1297)(192,1117)(102,1117)
\path(102,1117)(3252,1117)(3252,1792)
	(192,1297)(192,1117)(102,1117)
\blacken\path(57,1342)(192,1342)(192,1117)
	(57,1117)(57,1342)
\path(57,1342)(192,1342)(192,1117)
	(57,1117)(57,1342)
\blacken\path(3207,712)(3342,712)(3342,37)
	(3207,37)(3207,712)
\path(3207,712)(3342,712)(3342,37)
	(3207,37)(3207,712)
\blacken\path(3207,1792)(3342,1792)(3342,1117)
	(3207,1117)(3207,1792)
\path(3207,1792)(3342,1792)(3342,1117)
	(3207,1117)(3207,1792)
\blacken\path(12,712)(147,712)(147,487)
	(12,487)(12,712)
\path(12,712)(147,712)(147,487)
	(12,487)(12,712)
\end{picture}
}
\end{minipage}
\begin{minipage}[t]{0.2\textwidth}
\setlength{\unitlength}{0.00034996in}
\begingroup\makeatletter\ifx\SetFigFont\undefined%
\gdef\SetFigFont#1#2#3#4#5{%
  \reset@font\fontsize{#1}{#2pt}%
  \fontfamily{#3}\fontseries{#4}\fontshape{#5}%
  \selectfont}%
\fi\endgroup%
{\renewcommand{\dashlinestretch}{30}
\begin{picture}(3354,1819)(0,-10)
\texture{44555555 55aaaaaa aa555555 55aaaaaa aa555555 55aaaaaa aa555555 55aaaaaa 
	aa555555 55aaaaaa aa555555 55aaaaaa aa555555 55aaaaaa aa555555 55aaaaaa 
	aa555555 55aaaaaa aa555555 55aaaaaa aa555555 55aaaaaa aa555555 55aaaaaa 
	aa555555 55aaaaaa aa555555 55aaaaaa aa555555 55aaaaaa aa555555 55aaaaaa }
\shade\path(99,687)(3249,687)(3249,12)
	(99,462)(99,687)
\path(99,687)(3249,687)(3249,12)
	(99,462)(99,687)
\shade\path(147,487)(3252,1117)(3252,1792)
	(57,667)(57,487)(147,487)
\path(147,487)(3252,1117)(3252,1792)
	(57,667)(57,487)(147,487)
\shade\path(102,1117)(3252,1117)(3252,1792)
	(192,1297)(192,1117)(102,1117)
\path(102,1117)(3252,1117)(3252,1792)
	(192,1297)(192,1117)(102,1117)
\blacken\path(57,1342)(192,1342)(192,1117)
	(57,1117)(57,1342)
\path(57,1342)(192,1342)(192,1117)
	(57,1117)(57,1342)
\shade\path(192,1117)(3252,37)(3207,667)
	(192,1342)(192,1117)(192,1117)
\path(192,1117)(3252,37)(3207,667)
	(192,1342)(192,1117)(192,1117)
\blacken\path(3207,712)(3342,712)(3342,37)
	(3207,37)(3207,712)
\path(3207,712)(3342,712)(3342,37)
	(3207,37)(3207,712)
\blacken\path(3207,1792)(3342,1792)(3342,1117)
	(3207,1117)(3207,1792)
\path(3207,1792)(3342,1792)(3342,1117)
	(3207,1117)(3207,1792)
\blacken\path(12,712)(147,712)(147,487)
	(12,487)(12,712)
\path(12,712)(147,712)(147,487)
	(12,487)(12,712)
\end{picture}
}
\end{minipage}
\caption{Reachpipe Completion (i) $R(t) $; (ii) $\ceil{R}(t)$}
\label{figure:ReachpipeCompletion}
\end{figure}
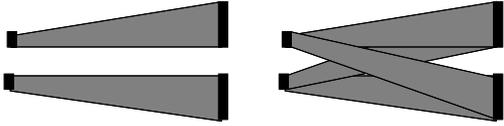
For an example, see Figure~\ref{figure:ReachpipeCompletion}, where
$R(t_k) \subseteq \reals$ and $ R(t_{k+1})\subseteq \reals$ are the 
disjoint black line segments
at the ends, and the shaded portions are the completions for $t\in (t_k, t_{k+1})$.
The left side shows $R(t)$.
The traces evolve from the top (resp. bottom) left black bars to the top (resp. bottom)
right black bars.
The figure on the right shows that   $\ceil{R}$ over-approximates by
assuming traces from the top left black bar to the bottom right black bar
(and similarly from the bottom left bar).
The strict inclusion can hold even if $R(t_k) $ and $ R(t_{k+1})$ are convex sets.

\smallskip\noindent\textbf{Reachpipe Sample Sets.}
We now look at choosing appropriate forms of 
reachpipe sample sets $R(t_k)$.
In hybrid systems literature the common forms of reach sets are
(i)~ellipsoids~\cite{KurzhanskiV06}, 
(ii)~support functions~\cite{GuernicG10},
(iii)~zonotopes~\cite{Girard05,GirardG08},
(iv)~polyhedra and 
polytopes~\cite{FrehseGDCRLRGDM11,HanK06,ChutinanK03,SankaranarayananDI08,SankaranarayananDI08,ColonS11},
(v)~polynomial approximations~\cite{PrabhakarV11, ChenAS12}.

In this work we use  convex polytopes as reachpipe sample sets. 
A \emph{polyhedron} is specified as:
$A\cdot \bx \leq  \bb$, where
$A$ is a $n \times d$ real-valued matrix, 
$\bx = [x_1, \dots, x_d]^{\transpose}$ is a column vector of $d$ variables,
$\bb =  [b_1, \dots, b_d]^{\transpose}$ is a column vector with $b_k\in \reals$ for 
every $k$, and ``$\cdot$" denotes the standard matrix product.
The polyhedron $A\cdot \bx \leq  \bb$ consists of all points
$(p_1, \dots, p_d) \in \reals^d$ such that for all $1\leq i \leq n$, we have 
$\sum_{k=1}^d A_{i,k}\cdot p_k \leq b_k$.
A polyhedron is thus the intersection of $n$ halfspaces, namely, the halfspaces
$\sum_{k=1}^d A_{i,k}\cdot x_k \leq b_k$ for $1\leq i \leq n$.
We use $\ba_i\cdot \bx \leq b_i$ as a shorthand to denote the $i$-th halfspace, where
$\ba_i$ is the $i$-th row vector of $A$.
A \emph{polytope} is a bounded polyhedron.
Polytopes can also be specified as convex hulls of  a finite set of 
points~\cite{ZieglerBook}
(unfortunately, polynomial time algorithms are not known to obtain  one representation
from the other~\cite{WelzlABS97}).
We use the halfspace representation as  it has been shown to be
 amenable to computing
over-approximations of reach sets of hybrid systems using the template polyhedra 
approach~\cite{HanK06,ChutinanK03,SankaranarayananDI08,SankaranarayananDI08,ColonS11}, in which
the reachsets at  sampled  timepoints are over-approximated by polytopes
by varying the constants in $\bb$ 
(the matrix $A$ stays unchanged).
Zonotopes are special forms of polytopes, the algorithms developed in this work are also
applicable for these special polytopes.

We note the property that if  $R(t_k)$ and $R(t_{k+1})$ are polytopes (resp. zonotopes)
in Equation~\eqref{equation:CompletionOver}, 
the completions $\ceil{R}(t)$ for every $t$ are also  polytopes (resp. zonotopes).
This follows from the facts that for $P_1$ and $P_2$ polytopes (resp. zonotopes),
(i)~$\lambda\cdot P_1$ and  $\lambda\cdot P_2$ are  polytopes (resp. zonotopes)
for $\lambda$ a constant; and
(ii)~the Minkowski sum $P_1 + P_2$ is also a polytope (resp. zonotope)~\cite{ZieglerBook}.

\smallskip\noindent\textbf{Polygonal Polytope-Reachpipe (\ppr).}
A \emph{polygonal polytope-reachpipe} (\ppr) is a reachpipe specified
by reachpipe time-samples $R(0),\dots R(m)$, such that for $k\in \set{0,1,\dots m-1}$ 
(a)~each $R(k)$ is a polytope in $\reals^{d+1}$; and
(b)~$R(t)$ for $k < t < k+1$  is taken to be
 the linear interpolation as specified in Equation~\eqref{equation:CompletionOver}.
Note that we take the reachpipe samples to occur at integer parameter values,
this is WLOG as the actual time value can be  added as an extra dimension as
discussed in Subsection~\ref{subsection:MovingFrechet} with a slight
modification: for a polygonal trace $f$ consisting of affine
segments 
starting at times $t_0, t_1,\dots$,
we let the corresponding (polygonal) time-explicit
trace $C$ be such that $C(k) = \left(f(t_k), t_k\right)$ for $k\in \set{0,1,\dots m}$ 
(for non-integer $\rho \in [0, m]$, the trace $C$ is specified by linear interpolation
of the integer endpoints).
Next,  we study the variation distance between  time-explicit \ppr{s}
with respect to the \frechet trace metric in order to bound the Skorokhod distance between
the corresponding tracepipes.

\section{Fr{\'e}chet  Distances between Polytope-Reachpipes}
\label{section:FrechetPipe}

We now investigate computing the pipe variation distance
bounds given in Proposition~\ref{proposition:OverUnder}
in the case of the Skorokhod trace metric.
As a first step, we show it suffices to consider the \frechet metric
as the trace metric in the pipe variation distance.


Consider the setting of  Subsection~\ref{subsection:ReachpipeConstruction}, which
presented linear interpolation completion of sampled trace values.
The traces so obtained by completion are continuous.
We can define corresponding time-explicit traces $C_{f}: [T_i^f, T_e^f] \rightarrow
\reals^d\times \reals$ for the 
 traces $f: [T_i^f, T_e^f] \rightarrow \reals^d$ obtained by completing the 
time sampled traces by linear interpolation.
This makes Proposition~\ref{proposition:SkoroToFrechet} applicable.
Corresponding to a tracepipe $F$ over $\reals^d$, we can define a 
time-explicit tracepipe $F^*$ over  $\reals^d\times \reals$ with traces $f\in F$ corresponding
to time-explicit traces $C_f$ in $F^*$.
We then have (referring to  trace metrics $\skoro$ or $\fre$ explicitly in
the variation distance through the notation 
 ${\dist_{\skoro}}_{\var}$ or ${\dist_{\fre}}_{\var}$):
\[
\begin{array}{ll}
{\dist_{\skoro}}_{\var}(F_1, F_2)  & = \sup_{f_1\in F_1, f_2\in F_2} \dist_{\skoro}(f_1, f_2)\\
& =  \sup_{\curve_{f_1}\in F_1^*, \curve_{f_2}\in F_2^*} \dist_{\fre}(\curve_{f_1},\curve_{f_2})\\
& = {\dist_{\fre}}_{\var}(F_1^*, F_2^*) 
\end{array}
\]
Thus we focus on computing the pipe variation distances with respect to the \frechet trace metric.

In Section~\ref{section:Pipes}, we considered distances between \emph{sets} of traces,
and investigated bounding the variation distance between sets of traces (\ie, between 
tracepipes) using over-approximate tracesets obtained through reachpipes.
In the next two subsections, we define a notion of
\frechet distance \emph{directly} on reachpipes, by viewing a reachpipe as a 
trace from $[0,T]$ to polytopes of $\reals^{d+1}$.
%

Let $R_1, R_2$ be \pprs from $[0, m_1]$ and $[0, m_2]$ to polytopes over  $\reals^{d+1}$.
Our objective is to 
bound the tracepipe variation distance with respect to the \frechet trace  metric.
From Proposition~\ref{proposition:OverUnder}, we
need to compute 
(a)~$\dvarfre\left(\fpipe(R_1), \fpipe(R_2)\right)$
and 
(b)~$\dminfre\left(\fpipe(R_1), \fpipe(R_2)\right)$.

\subsection{Variation Distance on PPRs}
\label{subsection:PipeToPoly-One}

In this subsection, we consider $\dvarfre\left(\fpipe(R_1), \fpipe(R_2)\right)$.
Recall that this value is defined as:
\begin{equation}
\label{equation:TracepipeRef}
\mspace{-0mu}\dvarfre\!\left(\fpipe(R_1), \fpipe(R_2)\right) =
 \sup_{f_1\in \fpipe(R_1), f_2\in \fpipe(R_2)} \!\!\dist_{\fre}(f_1, f_2)
\end{equation}
We define a new variation distance on reachpipes as follows.
\begin{definition}
\label{definition:ReachpipeFre}
Let $R_1, R_2$ be \pprs from $[0, m_1]$ and $[0, m_2]$ to
polytopes over $\reals^{d+1}$, and let $L$ be a given norm on
$\reals^{d+1}$.
The reachpipe variation distance
$\dvarfre^{\dagger}\left(R_1, R_2\right) $ is defined as:
\begin{equation}
\label{equation:ReachpipeFre}
 \inf_{\substack{
\alpha_1: [0,1] \rightarrow [0, m_1] \\
\alpha_2: [0,1] \rightarrow [0, m_2] }}
\max_{0\leq \theta \leq 1} \quad
\max_{
\substack{\bp_1\in R_1\left(\alpha_1(\!\theta\!)\right)\\
\bp_2\in R_2\left(\alpha_2(\!\theta\!)\right)}}
\norm{\bp_1-\bp_2}_{L}
\end{equation}
where $\alpha_1, \alpha_2$ range over  continuous and strictly increasing bijective functions
onto  $[0, m_1]$ and $[0, m_2]$ respectively.\qed
\end{definition}
Note that $\dvarfre^{\dagger}$ is defined over \emph{reachpipes} $R$,
as compared to $\dvarfre$ which is defined over tracepipes
$F$ or $\fpipe(R)$.
Also note that for any reparameterizations $\alpha_1, \alpha_2$, the sets 
$R_1\left(\alpha_1(\theta)\right)$ and $R_2\left(\alpha_2(\theta)\right)$
are closed and bounded.
Thus, $\max_{\bp_1\in R_1\left(\alpha_1(\!\theta\!)\right),\ 
\bp_2\in R_2\left(\alpha_2(\!\theta\!)\right)}
\norm{\bp_1-\bp_2}_{L}$ is well defined.
The function $\dvarfre^{\dagger}$, like the function $\dvarfre$, is not a metric
(notably, we can have $\dvarfre^{\dagger}(R,R) > 0$).

Informally, we go along the \pprs $R_1$ and $R_2$ according to our chosen
reparameterizations $\alpha_1, \alpha_2$, and compare the \emph{polytopes}
$R_1\left(\alpha_1(\theta)\right)$ and $R_2\left(\alpha_2(\theta)\right)$
for each value of $0\leq \theta\leq 1$. 
If we view a \ppr $R$ as a mapping from $[0, m]$ to the set of
polytopes of $\reals^{d+1}$, then 
Definition~\ref{definition:ReachpipeFre} seems similar to the 
definition of the \frechet distance over 
traces (Definition~\ref{def:frechet}), where we use the following function to compare
polytopes $P_1, P_2$:
\begin{equation}
\label{equation:PhiMax}
  \Phi_{\max}(P_1, P_2) = \max_{\bp_1\in P_1, \bp_2\in P_2}\norm{\bp_1- \bp_2}_L
\end{equation}
Using $\Phi_{\max}$, 
Equation~\eqref{equation:ReachpipeFre} can be written as:
\begin{equation}
\label{equation:PhiMaxReframe}
\dvarfre^{\dagger}\!\left(R_1, R_2\right) =
 \inf_{\substack{
\alpha_1: [0,1] \rightarrow [0, m_1] \\
\alpha_2: [0,1] \rightarrow [0, m_2] }}
\max_{0\leq \theta \leq 1} \Phi_{\max}\Big(
R_1\!\left(\alpha_1(\!\theta\!)\right),\,
R_2\!\left(\alpha_2(\!\theta\!)\right)
\!\!\!\Big)
\end{equation}

The following theorem shows that $\dvarfre^{\dagger}\!\left(R_1, R_2\right)$
over-approximates the tracepipe distance
$\dvarfre\left(\fpipe(R_1), \fpipe(R_2)\right)$.
\begin{theorem}
\label{theorem:ReachpipeFreVar}
Let $R_1, R_2$ be \pprs from $[0, m_1]$ and $[0, m_2]$ to
polytopes over $\reals^{d+1}$, and let $L$ be a given norm on
$\reals^{d+1}$.
We have
\[
\dvarfre^{\dagger}\!\left(R_1, R_2\right) \ \geq \ 
\dvarfre\big(\fpipe(R_1),\,  \fpipe(R_2)\big)
\]
where the tracepipe distance 
$\dvarfre\big(\fpipe(R_1),\,  \fpipe(R_2)\big)$
is as defined in Equation
\eqref{equation:TracepipeRef},
and the reachpipe distance $\dvarfre^{\dagger}\!\left(R_1, R_2\right) $
is as defined in Definition~\ref{definition:ReachpipeFre}.
\end{theorem}
\begin{proof}
Consider any $f_1\in  \fpipe(R_1)$, and any  $f_2\in  \fpipe(R_2)$.
We have 
\[
\mspace{-20mu}
\dist_{\fre}(f_1, f_2) =  \inf_{\substack{
\alpha_1: [0,1] \rightarrow [0, m_1] \\
\alpha_2: [0,1] \rightarrow [0, m_2] }}\ 
\max_{0\leq \theta \leq 1} \norm {f_1\left(\alpha_1(\theta)\right) - f_2\left(\alpha_2(\theta)\right)}_L
\]
Observe that
$f_j\left(\alpha_j(\theta)\right) \in R_j\left(\alpha_j(\theta)\right)$ for $j\in\set{1,2}$.
Thus, for every $\alpha_1, \alpha_2, \theta$, 
\[
\norm {f_1\left(\alpha_1(\theta)\right) - f_2\left(\alpha_2(\theta)\right)}_L
\leq 
 \Phi_{\max}\left(
R_1\left(\alpha_1(\theta)\right), \, R_2\left(\alpha_2(\theta)\right)
\right)
\]
Thus, we have
\[
\dist_{\fre}(f_1, f_2) \leq  \inf_{\substack{
\alpha_1: [0,1] \rightarrow [0, m_1] \\
\alpha_2: [0,1] \rightarrow [0, m_2] }}\ 
\max_{0\leq \theta \leq 1}
\Phi_{\max}\Big(
R_1\!\left(\alpha_1(\!\theta\!)\right),\,
R_2\!\left(\alpha_2(\!\theta\!)\right)
\Big)
\]
That is, for every $f_1\in  \fpipe(R_1)$ and  $f_2\in  \fpipe(R_2)$,
we have
$\dist_{\fre}(f_1, f_2)  \leq \dvarfre^{\dagger}\!\left(R_1, R_2\right)$.
This implies that 
$\sup_{f_1\in \fpipe(R_1),\, f_2\in \fpipe(R_2)} \dist_{\fre}(f_1, f_2)
\leq \dvarfre^{\dagger}\!\left(R_1, R_2\right)$.
%
%
\qedhere
\end{proof}
The above theorem can be applied with  $R_1=\ceil{\rpipe(F_1)}$ and
 $R_2=\ceil{\rpipe(F_2)}$ in order to obtain the upper bound in 
Proposition~\ref{proposition:OverUnder} using the reachpipe
variation distance $\dvarfre^{\dagger}$ between $\ceil{\rpipe(F_1)}$ and
$\ceil{\rpipe(F_2)}$.
We next consider the lower bound.

\subsection{Minimum Distance on PPRs}
\label{subsection:PipeToPoly-Two}

We now consider $\dminfre\left(F^{R_1}, F^{R_2}\right)$
for \pprs $R_1, R_2$  from $[0, m_1]$ and $[0, m_2]$ to
polytopes over  $\reals^{d+1}$ respectively.
This distance is defined as:
\begin{equation}
\label{equation:TracepipeMinRef}
\dminfre\big(\fpipe(R_1),\,  \fpipe(R_2)\big) =
 \inf_{f_1\in \fpipe(R_1),\,  f_2\in \fpipe(R_2)} \dist_{\fre}(f_1, f_2)
\end{equation}
Analogous to the $\dvarfre$ function of
Definition~\ref{definition:ReachpipeFre},
we define a minimum set distance  $\dminfre$ over reachpipes.
We use the following function to compare polytopes
(given a norm $L$ over $\reals^{d+1}$):
\begin{equation}
\label{equation:PhiMin}
  \Phi_{\min}(P_1, P_2) = \min_{\bp_1\in P_1,\, \bp_2\in P_2}\norm{\bp_1- \bp_2}_L
\end{equation}
Using this function, we define $\dminfre$ as follows.
\begin{definition}
\label{definition:ReachpipeMinFre}
Let $R_1, R_2$ be \pprs from $[0, m_1]$ and $[0, m_2]$ to
polytopes over $\reals^{d+1}$, and let $\Phi_{\min}$ be the
polytope comparison function as described previously.
The reachpipe minimum set distance
$\dminfre^{\dagger}\left(R_1, R_2\right) $ is defined as:
\begin{equation}
\dminfre^{\dagger}\!\left(R_1, R_2\right) =
\! \!\inf_{\substack{
\alpha_1: [0,1] \rightarrow [0, m_1] \\
\alpha_2: [0,1] \rightarrow [0, m_2] }}
\max_{0\leq \theta \leq 1} \Phi_{\min}\Big(
R_1\!\left(\alpha_1(\!\theta\!)\right),\,
R_2\!\left(\alpha_2(\!\theta\!)\right)
\!\!\Big)
\end{equation}
where $\alpha_1, \alpha_2$ range over  continuous and strictly increasing bijective functions
onto  $[0, m_1]$ and $[0, m_2]$ respectively.\qed
\end{definition}
The following theorem shows that $\dminfre^{\dagger}\!\left(R_1, R_2\right)$
is equal to the  tracepipe distance
$\dminfre\left(\fpipe(R_1), \fpipe(R_2)\right)$.
The proof of the theorem can be found in the Appendix.
\begin{theorem}
\label{theorem:ReachpipeFreMin}
Let $R_1, R_2$ be \pprs from $[0, m_1]$ and $[0, m_2]$ to
polytopes over $\reals^{d+1}$, and let $L$ be a given norm on
$\reals^{d+1}$.
We have
\[
\dminfre^{\dagger}\!\left(R_1, R_2\right) \ = \ 
\dminfre\big(\fpipe(R_1), \fpipe(R_2)\big)
\]
where the tracepipe distance $\dminfre\big(\fpipe(R_1), \fpipe(R_2)\big)$
is as defined in
Equation \eqref{equation:TracepipeMinRef},
and the reachpipe distance $\dminfre^{\dagger}\!\left(R_1, R_2\right)$
  is as defined in Definition~\ref{definition:ReachpipeMinFre}.
\qed
\end{theorem}
Theorems~\ref{theorem:ReachpipeFreVar} and~\ref{theorem:ReachpipeFreMin}
allow us to bound to the tracepipe variation distance $\dvarfre$ using the 
reachpipe distances $\dvarfre^{\dagger}$ and $\dminfre^{\dagger}$ that were defined
in the current section.
In the next section we present algorithms for computing these two reachpipe
distances over \pprs.

\section{Fr{\'e}chet  Distances between Polytope-Traces}
\label{section:ComputeFrechet}

Theorems~\ref{theorem:ReachpipeFreVar} and~\ref{theorem:ReachpipeFreMin} show that
the distance functions 
$\dvarfre^{\dagger}$ and $\dminfre^{\dagger}$ over PPRs can be used to bound
the tracepipe distances  $\dvarfre$ and $\dminfre$.
We now present procedures for computing
$\dvarfre^{\dagger}$ and $\dminfre^{\dagger}$ as follows.
In Subsection~\ref{subsection:FreeSpace} we extend the geometric 
free space concept 
used in \cite{AltG95,MajumdarP15} to compute the Fr\'echet distance between two traces to the case of 
\pprs, and show how the \ppr distance decision problem
can be reduced to a two-dimensional reachability problem.
In Subsection~\ref{subsection:Decision} we present algorithms for the
reachability problems corresponding to $\dvarfre^{\dagger}$ and $\dminfre^{\dagger}$.

\subsection{The Free Space for Polytope-Traces}
\label{subsection:FreeSpace}

Let $\ptopereals$ denote the set of all polytopes in $\reals^{d+1}$.
A \ppr $R$ defined over the time interval $[0,m]$
can be viewed as a polytope-trace, defined as a function
from $[0,m]$ to $\ptopereals$.
Recall that a \ppr $R$  is specified
by reachpipe time-samples $R(0),\dots R(m)$, such that for $k\in \set{0,1,\dots m-1}$
the portion of $R$ in between $(k, k+1)$ is assumed to be completed
according to linear interpolation using  $R(k)$
and $R(k+1)$. 
We denote this portion of $R$ between $R(k)$ and $R(k+1)$ as
$R^{[k]}$, \ie, the portion of $R$ defined over $k\leq t \leq k+1$.

Alt and Godau introduced \emph{free spaces} \cite{AltG95} to compute the \frechet distance between
piecewise affine and continuous curves in $\reals^d$. 
We show free spaces can also be used to compute the functions $\dvarfre^{\dagger}$ 
and $\dminfre^{\dagger}$.
First, we  show how to extend free spaces to the domain of \pprs.
\begin{definition}[Free Space]
\label{definition:FreeSpace}
Given \pprs
$R_1:[0, m_1] \rightarrow \ptopereals$ and $R_2:[0, m_2] \rightarrow \ptopereals$,  
a real number $\delta \geq 0$, 
and a polytope comparison function 
$\Phi: \ptopereals \times \ptopereals \rightarrow \reals_+$,
the \emph{$\delta$-Free Space} of $R_1, R_2$
with respect to $\Phi$
 is defined as  the set
$\free_{\delta}^{\Phi}(R_1,R_2)  = $
\[
\left\{(\rho_1, \rho_2) \in [0, m_1]\times [0,m_2]\,  \  \left\arrowvert\ 
\Phi\Big(
R_1(\rho_1), R_2(\rho_2)
\Big)
\leq 
\delta\right.\right\} \qed
\]
\end{definition}

The free space for \pprs serves a similar role as in the case of the free space for traces.
The tuples $(\rho_1, \rho_1) $ belonging to $\free_{\delta}^{\Phi}(R_1, R_2) $
denote the positions in the two reparameterizations such that
the $\Phi$ value for those position pairs is at most $\delta$.
Thus  $\free_{\delta}^{\Phi}(R_1, R_2) $ collects the pairs $(\rho_1, \rho_2) $ which could
be used in valid reparameterizations of Definition~\ref{definition:ReachpipeFre}
or~\ref{definition:ReachpipeMinFre}.
A pictorial representation of the free space is referred to as the
\emph{free space diagram}.
The space $[0, m_1]\times [0,m_2]$ can be viewed as consisting of
$m_1m_2$  \emph{cells}, with  cell $i,j$ being $[i,i+1]\times [j,j+1]$
for $0\leq i<m_1$, and $0\leq j<m_1$.
Observe that $\free_{\delta}^{\Phi}(R_1, R_2)$ intersected with cell $i,j$ is just the free space
corresponding to the \ppr segments  $R_1^{[i]}, R_2^{[j]}$;
\emph{i.e.,} the intersection of the cell $i,j$ with $\free_{\delta}^{\Phi}(R_1, R_2) $
is equal to $\free_{\delta}^{\Phi}(R_1^{[i]}, R_2^{[j]}) $.

\begin{proposition}[Free Space \&  Reparameterizations]
\label{proposition:FreeSpaceCurve}
Given two \pprs  $R_1,R_2$ from
$[0,m_1]$ and $[0,m_2]$ to $\ptopereals$, we have 
$\dvarfre^{\dagger}\!\left(R_1, R_2\right)\! \leq \! \delta$
(resp., $\dminfre^{\dagger}\!\left(R_1, R_2\right)\! \leq\! \delta$) iff
there is a non-decreasing (in both dimensions)
curve $\alpha: [0,1]\! \rightarrow\!  [0, m_1]\!\times\! [0,m_2]$
 in $\free_{\delta}^{\Phi_{\max}}(R_1,R_2) $
(resp. $\free_{\delta}^{\Phi_{\min}}(R_1,R_2) $)
 from $(0,0)$ to $(m_1, m_2)$.
 \qed
\end{proposition}
The curve $\alpha$ can be thought of as a pair of
parameterized curves $(\alpha_1, \alpha_2)$,
with
$\alpha_1:  [0,1] \rightarrow  [0, m_1]$ and $\alpha_2:  [0,1] \rightarrow  [0, m_2] $.
The functions $\alpha_1, \alpha_2$ can be viewed as the reparameterization functions
in Definitions~~\ref{definition:ReachpipeFre} and~\ref{definition:ReachpipeMinFre}.
The general shape of the free space for two \pprs is depicted in
Figure~\ref{figure:FreeSpace}. 
The unshaded portion is the free space.
The figure also includes a continuous curve which is
non-decreasing  in both coordinates, from $(0,0)$ to $(m_1, m_2)$.
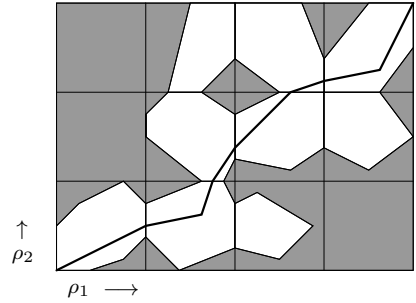
\begin{figure}[t]
\vspace{-1em}
\strut\centerline{\setlength{\unitlength}{0.00026247in}
\begingroup\makeatletter\ifx\SetFigFont\undefined%
\gdef\SetFigFont#1#2#3#4#5{%
  \reset@font\fontsize{#1}{#2pt}%
  \fontfamily{#3}\fontseries{#4}\fontshape{#5}%
  \selectfont}%
\fi\endgroup%
{\renewcommand{\dashlinestretch}{30}
\begin{picture}(8148,6088)(0,-10)
\path(915,4240)(8115,4240)
\path(2715,6040)(2715,640)
\path(4515,6040)(4515,640)
\path(6315,6040)(6315,640)
\texture{55888888 88555555 5522a222 a2555555 55888888 88555555 552a2a2a 2a555555 
	55888888 88555555 55a222a2 22555555 55888888 88555555 552a2a2a 2a555555 
	55888888 88555555 5522a222 a2555555 55888888 88555555 552a2a2a 2a555555 
	55888888 88555555 55a222a2 22555555 55888888 88555555 552a2a2a 2a555555 }
\shade\path(915,6040)(8115,6040)(8115,640)
	(915,640)(915,6040)
\path(915,6040)(8115,6040)(8115,640)
	(915,640)(915,6040)
\path(915,2440)(8115,2440)
\whiten\path(915,1540)(1365,1990)(2265,2440)
	(2715,1990)(2715,1315)(2265,865)
	(1590,640)(915,640)(915,1540)
\path(915,1540)(1365,1990)(2265,2440)
	(2715,1990)(2715,1315)(2265,865)
	(1590,640)(915,640)(915,1540)
\whiten\path(2715,1990)(3840,2440)(4290,2440)
	(4515,1990)(4515,1090)(3390,640)
	(2715,1315)(2715,1990)
\path(2715,1990)(3840,2440)(4290,2440)
	(4515,1990)(4515,1090)(3390,640)
	(2715,1315)(2715,1990)
\whiten\path(4515,1990)(4965,2215)(6090,1540)
	(5415,865)(4515,1090)(4515,1990)
\path(4515,1990)(4965,2215)(6090,1540)
	(5415,865)(4515,1090)(4515,1990)
\whiten\path(4515,3790)(5415,4240)(6315,4240)
	(6315,3115)(5640,2665)(4515,2890)(4515,3790)
\path(4515,3790)(5415,4240)(6315,4240)
	(6315,3115)(5640,2665)(4515,2890)(4515,3790)
\whiten\path(6315,3115)(7215,2665)(8115,3340)
	(7440,4240)(6315,4240)(6315,3115)
\path(6315,3115)(7215,2665)(8115,3340)
	(7440,4240)(6315,4240)(6315,3115)
\whiten\path(2715,3790)(2715,3340)(3840,2440)
	(4290,2440)(4515,2890)(4515,3790)
	(3840,4240)(3165,4240)(2715,3790)
	(2715,3340)(2715,3790)
\path(2715,3790)(2715,3340)(3840,2440)
	(4290,2440)(4515,2890)(4515,3790)
	(3840,4240)(3165,4240)(2715,3790)
	(2715,3340)(2715,3790)
\whiten\path(6315,4915)(7215,6040)(8115,6040)
	(8115,5140)(7440,4240)(6315,4240)(6315,4915)
\path(6315,4915)(7215,6040)(8115,6040)
	(8115,5140)(7440,4240)(6315,4240)(6315,4915)
\whiten\path(5415,4240)(4515,4915)(4515,6040)
	(5865,6040)(6315,4915)(6315,4240)(5415,4240)
\path(5415,4240)(4515,4915)(4515,6040)
	(5865,6040)(6315,4915)(6315,4240)(5415,4240)
\whiten\path(3165,4240)(3615,6040)(4515,6040)
	(4515,4915)(3840,4240)(3165,4240)
\path(3165,4240)(3615,6040)(4515,6040)
	(4515,4915)(3840,4240)(3165,4240)
\thicklines
\path(915,640)(2715,1540)(3840,1765)
	(4065,2440)(4515,3115)(5640,4240)
	(6315,4465)(7440,4690)(8115,6040)
\thinlines
\path(2715,640)(2715,6040)
\path(6315,6040)(6315,640)
\path(915,4240)(8115,4240)
\path(4515,640)(4515,6040)
\put(15,865){\makebox(0,0)[lb]{\smash{{\SetFigFont{8}{9.6}{\familydefault}{\mddefault}{\updefault}$\rho_2$}}}}
\put(1140,190){\makebox(0,0)[lb]{\smash{{\SetFigFont{8}{9.6}{\familydefault}{\mddefault}{\updefault}$\rho_1$}}}}
\put(1860,145){\makebox(0,0)[lb]{\smash{{\SetFigFont{8}{9.6}{\familydefault}{\mddefault}{\updefault}$\longrightarrow$}}}}
\put(105,1360){\makebox(0,0)[lb]{\smash{{\SetFigFont{8}{9.6}{\familydefault}{\mddefault}{\updefault}$\uparrow$}}}}
\end{picture}
}}
\vspace*{-6mm}
 \caption{The Free Space $\free_{\delta}^{\Phi}(R_1,R_2)$.}
\label{figure:FreeSpace}
\end{figure}

Note  that the curve $\alpha$ (and hence also each of $\alpha_1, \alpha_2$)
in  Proposition~\ref{proposition:FreeSpaceCurve} is non-decreasing;
whereas the reparameterizations in
 Definitions~\ref{definition:ReachpipeFre}
and~\ref{definition:ReachpipeMinFre} are strictly increasing.
This is to account for the fact that optimal reparameterizations
in  Definitions~\ref{definition:ReachpipeFre}
and~\ref{definition:ReachpipeMinFre} might not exist, as we have an ``$\inf$''.
It can be shown that 
$\dminfre^{\dagger} $ and $\dvarfre^{\dagger}$ values do not change over \pprs  if we
allow non-decreasing reparameterizations  since  \pprs change smoothly due to the
linear interpolation scheme.
This issue also arises in the case of traces, and is discussed (for the case of traces) 
in more detail
in~\cite{MajumdarP15}.
We omit the technicalities, and henceforth assume that non-decreasing
reparameterizations are allowed in 
 Definitions~\ref{definition:ReachpipeFre}
and~\ref{definition:ReachpipeMinFre}.

\subsection{The Polytope-Trace $\bm{\dist^{\dagger}}$  Decision Problems}
\label{subsection:Decision}

In this section, we  solve for the decision problems
$\dvarfre^{\dagger}(R_1, R_2) \leq \delta$ and $\dminfre^{\dagger}(R_1, R_2) \leq \delta$,
given a $\delta \geq 0$ and \pprs $R_1, R_2$.
We use the free space reduction of Proposition~\ref{proposition:FreeSpaceCurve}
for these decision problems.
The first step in this procedure is to compute the free space.
Towards this step, we first show that the free spaces for 
the polytope comparison functions $\Phi_{\min}$ and $\Phi_{\max}$
are convex in individual cells of the free space diagram.
This is done in Subsection~\ref{subsubsection:Convexity}.
Using this convexity property, we show in 
Subsection~\ref{subsubsection:ComputeFree} that 
in order to obtain the free space of a cell, it
suffices to obtain the free space at the cell boundaries.
We obtain algorithms to compute the free space cell boundaries
in Subsection~\ref{subsubsection:PhiMin} (for $\Phi_{\min}$), and
in~\ref{subsubsection:PhiMaxOne} (for $\Phi_{\max}$).
The procedure of Subsection~\ref{subsubsection:PhiMaxOne} has a high
time complexity, we present a polynomial time algorithm which works
in case the \pprs satisfy certain conditions in Subsection~\ref{subsubsection:PhiMaxPoly}.
The results of the section are summarized in 
Propositions~\ref{proposition:DecisionMin},~\ref{proposition:DecisionVarOne} 
and~\ref{proposition:DecisionVarPoly}.

\subsubsection{Convexity of Free Space}
\label{subsubsection:Convexity}

The following lemma proves that the free space in the first cell (over $[0, 1]\times [0,1]$)
is convex for both  the set comparison functions $\Phi_{\min}$ and $\Phi_{\max}$.
Other cells are translations
and have a similar proof.

\begin{lemma}[Convexity of Free Space of Individual Cells]
\label{lemma:Convexity}
Let $P_a^0, P_a^1$, and $P_b^0, P_b^1$ be polytopes in $\reals^{d+1}$.
Let $R_a: [0,1]\rightarrow \ptopereals$ and $R_b: [0,1]\rightarrow \ptopereals$
be (single-segment) \pprs constructed from  
the polytopes  $P_a^0, P_a^1$ and $P_b^0, P_b^1$ respectively, via linear interpolation
(as described in Equation~\eqref{equation:CompletionOver}),
taking $P_a^0 = R_a(0)$ and  $P_a^1 = R_a(1)$ and
$P_b^0 = R_b(0)$ $P_b^1 = R_b(1)$, respectively.

The free space of  $R_a, R_b$ given a $\delta \geq 0$
for both $\Phi_{\min}$ and  $\Phi_{\max}$ is convex.
That is,   $\free_{\delta}^{\Phi_{\min}}(R_a,R_b)$ and
$\free_{\delta}^{\Phi_{\max}}(R_a,R_b)$ are both convex sets.
\end{lemma}
\begin{proof}
Let $\Phi $ be  $\Phi_{\min}$ or $\Phi_{\max}$.
Suppose two points (in $[0,1]\times [0,1]$)
belong to  $\free_{\delta}^{\Phi}(R_a,R_b)$.
Let these points be $\rho= (\rho_a, \rho_b)$ and 
$\rho'= (\rho_a', \rho_b')$.
We show that for any $0\leq \lambda\leq 1$,
the point $\rho^* = \lambda \cdot \rho + (1-\lambda)\cdot \rho'$ also belongs
to  $\free_{\delta}^{\Phi}(R_a,R_b)$.
The point  $\rho^*$ is the tuple
\begin{equation}
\label{equation:ConvexityMin}
(\rho_a^*, \rho_b^*) = 
\big( \lambda \cdot \rho_a + (1-\lambda)\cdot \rho_a'\ , \ 
\lambda \cdot \rho_b + (1-\lambda)\cdot \rho_b'
\big).
\end{equation}
To show $(\rho_a^*, \rho_b^*) \in \free_{\delta}^{\Phi}(R_a,R_b)$, we need to show that
\begin{equation}
\label{equation:ConvexityPhi}
\Phi\big( R_a(\rho_a^*), R_b(\rho_b^*) \big) \leq \delta
\end{equation}
We show this individually for   $\Phi_{\min}$ and $\Phi_{\max}$.

\smallskip\noindent
\textbf{(1)}~$\bm{\Phi_{\min}}$.\\
By the definition of $\Phi_{\min}$ (Equation~\eqref{equation:PhiMin}), and the facts that
$(\rho_a, \rho_b)$ and $(\rho_a', \rho_b')$ are in $\free_{\delta}^{\Phi_{\min}}(R_a,R_b)$,
we have that:
\begin{compactitem}
\item  There exist points $\bp_a\in R_a(\rho_a)$ and $\bp_b\in R_b(\rho_b)$ such that
$\norm{\bp_a - \bp_b} \leq \delta$.
\item There exist points $\bp_a'\in R_a(\rho_a')$ and $\bp_b'\in R_b(\rho_b')$ such that
$\norm{\bp_a' - \bp_b'} \leq \delta$.
\end{compactitem}
Consider the points $\bp_a^* = \lambda \cdot \bp_a + (1-\lambda) \cdot \bp_a'$; and
$\bp_b^* = \lambda \cdot \bp_b + (1-\lambda) \cdot \bp_b'$
(where $\lambda$ is the same value as that used in
Equation~\eqref{equation:ConvexityMin}).
We have 
\begin{align*}
\mspace{-5mu}\norm{\bp_a^* - \bp_b^*} &  = 
\norm{ \Big(\lambda \cdot \bp_a + (1-\lambda) \cdot \bp_a'\Big) -
\Big( \lambda \cdot \bp_b + (1-\lambda) \cdot \bp_b'\Big)}\\
& = \norm{
\lambda\cdot \left( \bp_a-  \bp_b\right) +
(1-\lambda) \cdot \left( \bp_a'-  \bp_b'\right) 
}\\
& \leq  \lambda\cdot \norm{\bp_a-  \bp_b} + (1-\lambda) \cdot  \norm{\bp_a'-  \bp_b'}\\
&\mspace{180mu} \text{(by basic norm properties)}\\
& \leq \lambda \cdot \delta + (1-\lambda) \cdot   \delta\\
& = \delta
\end{align*}
We now show  $\bp_a^* \in R_a(\rho_a^*)$, and $\bp_b^* \in R_b(\rho_b^*)$
Observe  that the polytope $R_a(\rho_a^*)$ which is defined to be the polytope
\begin{align}
\label{equation:PolytopeInterpolate}
\begin{split}
& R_a(0) + \rho_a^*\cdot \left(R_a(0) - R_a(1)\right)\\
& = R_a(0) + \Big(\lambda \cdot \rho_a + (1-\lambda)\cdot \rho_a'\Big)\cdot
\left(R_a(0) - R_a(1)\right)\\
& = \lambda \cdot \left(R_a(0) +  \rho_a\cdot \left(R_a(0) - R_a(1)\right) \right)\ + \\
& \qquad\qquad \qquad\qquad
 (1-\lambda) \cdot \left(R_a(0) +  \rho_a'\cdot \left(R_a(0) - R_a(1)\right) \right)\\
& = \lambda \cdot R_a(\rho_a) + (1-\lambda) \cdot R_a(\rho_a') 
\end{split}
\end{align}
%
Thus, $R_a(\rho_a^*)$ equals the polytope 
$ \lambda \cdot R_a(\rho_a) + (1-\lambda) \cdot R_a(\rho_a') $.
Since $\bp_a^* = \lambda \cdot \bp_a + (1-\lambda) \cdot \bp_a'$ for
$\bp_a\in  R_a(\rho_a) $ and $ \bp_a'\in R_a(\rho_a') $,
this means that $\bp_a^* \in R_a(\rho_a^*)$.
Similarly, $\bp_b^* \in R_b(\rho_b^*)$.
Since we have demonstrated that $\norm{\bp_a^* - \bp_b^*} \leq \delta$, this means
that $\Phi_{\min}\left(R_a(\rho_a^*), R_b(\rho_b^*)\right) \leq \delta$.
This shows that Equation~\eqref{equation:ConvexityPhi} holds for $\Phi_{\min}$.

\smallskip\noindent
\textbf{(2)}~$\bm{\Phi_{\max}}$.\\
Now we show that  Equation~\eqref{equation:ConvexityPhi} holds for $\Phi_{\max}$.
By the definition of $\Phi_{\max}$ (Equation~\eqref{equation:PhiMax}), and the facts that
$(\rho_a, \rho_b)$ and $(\rho_a', \rho_b')$ are in $\free_{\delta}^{\Phi_{\min}}(R_a,R_b)$,
we have that:
\begin{compactitem}
\item  For all points $\bp_a\in R_a(\rho_a)$ and $\bp_b\in R_b(\rho_b)$ we have that 
$\norm{\bp_a - \bp_b} \leq \delta$.
\item For all points $\bp_a'\in R_a(\rho_a')$ and $\bp_b'\in R_b(\rho_b')$ we have that
$\norm{\bp_a' - \bp_b'} \leq \delta$.
\end{compactitem}
Consider any point $\bp_a^* $ which belongs to
$ R_a(\rho_a^*)$ and any point $\bp_b^* $ which belongs to $ R_b(\rho_b^*)$.
By Equation~\eqref{equation:PolytopeInterpolate}, we have
$ R_a(\rho_a^*) = \lambda \cdot R_a(\rho_a) + (1-\lambda) \cdot R_a(\rho_a') $;
and similarly for $ R_b(\rho_b^*) $
Thus, by definition, 
\begin{compactitem}
\item 
$\bp_a^* = \lambda \cdot \bp_a + (1-\lambda) \cdot \bp_a'$
for some $\bp_a\in R_a(\rho_a) $ and $\bp_a'\in R_a(\rho_a') $; and
\item  
$\bp_b^* = \lambda \cdot \bp_b + (1-\lambda) \cdot \bp_b'$
for some $\bp_b\in R_a(\rho_b) $ and $\bp_b'\in R_b(\rho_b') $
\end{compactitem}
It can be shown (as in the $\Phi_{\min}$ case)  
using the above two facts that
$\norm{\bp_a^* - \bp_b^* } \leq \delta$.
That is, we have that for any  point $\bp_a^* \in R_a(\rho_a^*)$, and any
point $\bp_b^* \in R_b(\rho_b^*)$, the value $\norm{\bp_a^* - \bp_b^* }$
does not exceed $\delta$.
This means that
\[
\sup_{\bp_a^* \in R_a(\rho_a^*), \, \bp_b^* \in R_b(\rho_b^*)}
\norm{\bp_a^* - \bp_b^* } \leq \delta
\]
Thus, $\Phi_{\max}\left( R_a(\rho_a^*), R_b(\rho_b^*)\right) \leq \delta$.
This shows that Equation~\eqref{equation:ConvexityPhi} holds also for $\Phi_{\max}$
(in addition to $\Phi_{\min}$).
\end{proof}

\subsubsection{Computing the Free Space}
\label{subsubsection:ComputeFree}

The convexity demonstrated  by Lemma~\ref{lemma:Convexity} simplifies the problem of
computing a non-decreasing curve in the free space.
As a result of the convexity of the free space for a cell, it suffices to only compute the
free space boundaries at the cell boundaries.
\begin{figure}[t]
\strut\centerline{\setlength{\unitlength}{0.00021872in}
\begingroup\makeatletter\ifx\SetFigFont\undefined%
\gdef\SetFigFont#1#2#3#4#5{%
  \reset@font\fontsize{#1}{#2pt}%
  \fontfamily{#3}\fontseries{#4}\fontshape{#5}%
  \selectfont}%
\fi\endgroup%
{\renewcommand{\dashlinestretch}{30}
\begin{picture}(6195,6768)(0,-10)
\texture{55888888 88555555 5522a222 a2555555 55888888 88555555 552a2a2a 2a555555 
	55888888 88555555 55a222a2 22555555 55888888 88555555 552a2a2a 2a555555 
	55888888 88555555 5522a222 a2555555 55888888 88555555 552a2a2a 2a555555 
	55888888 88555555 55a222a2 22555555 55888888 88555555 552a2a2a 2a555555 }
\shade\path(1320,5739)(5820,5739)(5820,1239)
	(1320,1239)(1320,5739)
\path(1320,5739)(5820,5739)(5820,1239)
	(1320,1239)(1320,5739)
\whiten\path(2220,1239)(1320,2139)(1320,3939)
	(2220,5289)(3570,5739)(4020,5739)
	(4470,5739)(4920,5739)(5820,4839)
	(5820,3039)(4920,1689)(3570,1239)(2220,1239)
\path(2220,1239)(1320,2139)(1320,3939)
	(2220,5289)(3570,5739)(4020,5739)
	(4470,5739)(4920,5739)(5820,4839)
	(5820,3039)(4920,1689)(3570,1239)(2220,1239)
\path(2220,1014)(2220,1419)
\path(3570,1014)(3570,1419)
\path(6045,3039)(5640,3039)
\path(6045,4839)(5595,4839)
\path(1095,3939)(1545,3939)
\path(1095,2139)(1545,2139)
\path(3570,5964)(3570,5559)
\path(4920,5964)(4920,5559)
\thicklines
\dottedline{240}(1320,3444)(3570,5739)
\dottedline{240}(2985,1239)(4920,5739)
\put(3435,6234){\makebox(0,0)[lb]{\smash{{\SetFigFont{10}{12.0}{\familydefault}{\mddefault}{\updefault}$\mya^2$}}}}
\put(4830,6234){\makebox(0,0)[lb]{\smash{{\SetFigFont{10}{12.0}{\familydefault}{\mddefault}{\updefault}$\myb^2$}}}}
\put(6180,4794){\makebox(0,0)[lb]{\smash{{\SetFigFont{10}{12.0}{\familydefault}{\mddefault}{\updefault}$\myb^1$}}}}
\put(6135,2994){\makebox(0,0)[lb]{\smash{{\SetFigFont{10}{12.0}{\familydefault}{\mddefault}{\updefault}$\mya^1$}}}}
\put(3615,204){\makebox(0,0)[lb]{\smash{{\SetFigFont{10}{12.0}{\familydefault}{\mddefault}{\updefault}$\myb^0$}}}}
\put(2085,249){\makebox(0,0)[lb]{\smash{{\SetFigFont{10}{12.0}{\familydefault}{\mddefault}{\updefault}$\mya^0$}}}}
\put(15,3984){\makebox(0,0)[lb]{\smash{{\SetFigFont{10}{12.0}{\familydefault}{\mddefault}{\updefault}$\myb^3$}}}}
\put(15,2139){\makebox(0,0)[lb]{\smash{{\SetFigFont{10}{12.0}{\familydefault}{\mddefault}{\updefault}$\mya^3$}}}}
\end{picture}
}}
\vspace*{-2em}
 \caption{\hspace*{-1mm}Cell Crossing with Non-Decreasing Curves.}
\label{figure:Cell}
\end{figure}
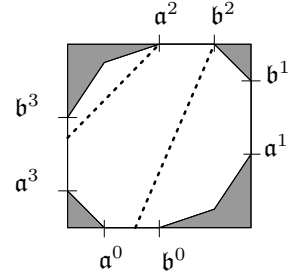
We refer to Figure~\ref{figure:Cell}.
The dotted lines are example non-decreasing curves that  cross the cell.
As can be seen, to check if we can 
go from the left free space boundary to the top free space boundary of the cell,
we only need the  top free space boundary (and the precondition that the left free
space boundary is non-empty).
A similar situation arises for checking traversal from the bottom to top or bottom to right
boundaries via non-decreasing curves.
Convexity makes the internal shape of the free space inside a cell irrelevant.
Invoking convexity again, we actually only need to compute the
points $\mya^k, \myb^k$ for $k\in\set{0,3}$.
We present the computation procedure next.

We compute the bottom free space boundaries of cells (the other boundaries 
have similar algorithmic solutions).
We need to compute the points $\mya^0, \myb^0$ in Figure~\ref{figure:Cell}.
We do this for the first cell (over $[0, 1]\times [0,1]$), other cells are translations and are
similar.
The point $\mya^0 = \tuple{\lambda^{\min}, 0}$, and the 
point $\myb^0 = \tuple{\lambda^{\max}, 0}$ for some $\lambda^{\min}$ and
$ \lambda^{\max}$ in $[0,1]$.
It hence suffices to  compute $ \lambda^{\min}$ and $ \lambda^{\max}$.
We solve for $ \lambda^{\min}$ (the solution for  $ \lambda^{\max}$ is similar) . 
This value  $ \lambda^{\min}$  is the solution of the following optimization problem
(where $R_1(0), R_1(1), R_2(0)$ are given polytope samples of \pprs $R_1$ and $R_2$) :
 { 
\begin{alignat*}{2}
  \text{minimize  }\  &\  \lambda\\
  \text{subject to }\  &
 \Phi\left(R_1(\lambda), R_2(0)\right) \leq \delta\\
 & 0 \leq  \lambda \leq 1
   \end{alignat*}
}
Expanding $R_1(\lambda)$, we get:
 {  
\vspace*{-1mm}
\begin{equation}
\label{opt:Phi}
\begin{alignedat}{2} 
  \text{minimize  }\  &\  \lambda\\
  \text{subject to }\  &
 \Phi\big(\lambda\vdot R_1(0) + (1-\lambda) \vdot R_1(1),\  R_2(0)\big) \leq \delta\\ 
 & 0 \leq  \lambda \leq 1
\end{alignedat}
\end{equation}}

The solution to the above problem depends on the function $\Phi$.
We solve each case $\Phi_{\min}$ and $\Phi_{\max}$ individually.


\subsubsection{Free Space Cell Boundaries for $\bm{\Phi_{\min}}$}
\label{subsubsection:PhiMin}

\noindent In this subsection, we compute the  bottom free space boundary of the first
cell (over $[0, 1]\times [0,1]$).
The optimization problem~\eqref{opt:Phi}  
for $\Phi= \Phi_{\min}$ has the same solution as:
{ 
\vspace*{-1mm}
\begin{alignat*}{2} 
  \text{minimize  }\  &\  \lambda\\[-1mm]
  \text{such that }\  &
\begin{array}{l}
\exists \text{ point } \bp\in   \lambda\cdot R_1(0) + (1-\lambda) \cdot R_1(1),\\
\exists \text{ point } \bq\in  R_2(0)
\end{array}\\
& \qquad \qquad \text{ s.t. } \norm{\bp - \bq} \leq \delta\\
 & 0 \leq  \lambda \leq 1
   \end{alignat*}
}
Let $R_1(0)$ be the polytope $A_1^0\cdot \bx \leq \bb_1^0$,
$R_1(1)$ be the polytope $A_1^1\cdot \bx \leq \bb_1^1$, and
$R_2(0)$ be the polytope $A_2\cdot \bx \leq \bb_2$;
where the $A$s are $n\times (d+1)$ matrices of given constants, and
$\bb$s are column vectors  of size $d+1$ containing given constants;
and $\bx$s  are column vectors of  variables.
The previous optimization problem can be stated using these polytopes as:
{ 
\begin{equation}
\label{opt:MinOrig}
\begin{alignedat}{2} 
  \text{minimize  }\  &\  \lambda\\[-1mm]
  \text{subject to }\  &
\norm{\lambda \cdot \bx^0 + (1-\lambda) \cdot \bx^1 - 
\by} \leq \delta\\ 
& A_1^0\cdot \bx^0 \leq \bb_1^0\\
& A_1^1\cdot \bx^1 \leq \bb_1^1\\
& A_2\cdot \by \leq \bb_2\\
& 0 \leq  \lambda \leq 1
\end{alignedat}
\end{equation}
}
The optimization above is over the variables $\lambda, \bx^0, \bx^1, \by$.
The values for $A_1^0, A_1^1, A_2, \bb_1^0, \bb_1^1, \bb_2, \delta$ are given.
We would like to reduce the problem to Linear Programming (LP), however we note that,
as stated, the problem is an instance of quadratic programming due to the
multiplication of the parameter $\lambda$ with parameter column 
vectors  $ \bx^0$ and $ \bx^1$.
We show  that these multiplicative constraints can be removed.
Towards this, we need the following lemma.
\begin{lemma}
\label{lemma:PolytopeZero}
Suppose $A\cdot \bx \leq \bb$ is a non-empty polytope in $\reals^{d+1}$ and
$\bb \neq \bm{0}$.
Then $A\cdot \bx \leq \bm{0}$ either has no solution, or  contains the only point
 $\bx =  \bm{0}$.
\qed
\end{lemma}

Using the above lemma, the following result can be shown 
(the proof is in the Appendix).
\begin{lemma}
\label{lemma:MinNew}
Let $A_1^0\cdot \bx^0 \leq \bb_1^0$, and $A_1^1\cdot \bx^1 \leq \bb_1^1$, and
$A_2\cdot \by \leq \bb_2$  be non-empty polytopes in $\reals^{d+1}$.
The following optimization problem has the same solution as
Problem~\eqref{opt:MinOrig}.
{ 
\begin{equation}
\label{opt:MinNew}
\begin{alignedat}{2} 
  \text{minimize  }\  &\  \lambda\\[-1mm]
  \text{subject to }\  &
\norm{\bz^0 + \bz^1  - 
\by} \leq \delta\\ 
& A_1^0\cdot \bz^0 \leq \lambda\cdot \bb_1^0\\
& A_1^1\cdot \bz^1 \leq (1-\lambda)\cdot \bb_1^1\\
& A_2\cdot \by \leq \bb_2\\
& 0 \leq  \lambda \leq 1 \qquad\qquad\qquad\qquad\qquad\qed
\end{alignedat}
\end{equation}
}
\end{lemma}
We thus can take $\lambda_{\min}$ to be the solution of the optimization 
problem~\eqref{opt:MinNew}.
Consider the norms $L_1^{\max}$ 
(recall the derived norms   given in Equation~\eqref{equation:LMaxDef});  or 
$L_{\infty}^{\max} $ (which is just the same as the
$L_{\infty}$ norm). 
Let us use any of  these  norms as the  norm in  $\norm{\bz^0 + \bz^1  - 
\by}$.
The optimization problem~\eqref{opt:MinNew} as stated is not a LP instance.
However, we showed in~\cite{MajumdarP15} how constraint problems involving
the $L_1^{\max}$, or $L_{\infty}$ norms can be framed as LP by 
doubling the number of variables.
A similar approach works here, thus, Problem~\eqref{opt:MinNew} can be solved using
linear programming.
We solved for the minimal $\lambda$. 
We can employ the same techniques for finding the maximal $\lambda$.
This gives us the following result.
\begin{proposition}[Free Space Cell Boundaries for $\Phi_{\min}$]
Given two \pprs $R_1, R_2$,  the set $\free_{\delta}^{\Phi_{\min}}(R_1,R_2)$ at 
cell-$(i,k)$ boundaries can be computed in time
$O\left(\LP\left(S_1^i+S_1^{i+1}+S_2^k+S_2^{k+1}\right)\right)$, where
$S_j^l$ denotes the halfspace representation size of polytope
$R_j(l)$, and $\LP(\cdot)$ is the (polynomial time) upper bound for solving
linear programming instances.\qed
\end{proposition}
After computing the free space cell boundaries, we can employ a dynamic programming
algorithm to check if there is a non-decreasing curve travelling through the free space 
from the point $(0,0)$ to $(m_1, m_2)$.
\begin{proposition}[$\dminfre^{\dagger}$ Decision Problem]
\label{proposition:DecisionMin}
Given  \pprs $R_1, R_2$ represented as $m_1$, $m_2$  polytopes respectively ,
and a $\delta \geq 0$, we can decide the question
$\dminfre^{\dagger}(R_1, R_2) \leq \delta$ in time 
$O\left(m_1\vdot m_2\vdot \LP(S_{\max})\right)$
for both $L_1^{\max}$ and  $L_{\infty}$  norms on $\reals^{d+1}$,
where
$S_{\max}$ is the maximum of the halfspace representation sizes of the given
polytopes, and $\LP(\cdot)$ is the (polynomial time) upper bound for solving
linear programming.
\qed
\end{proposition}

\subsubsection{Free Space Cell Boundaries for $\bm{\Phi_{\max}}$}
\label{subsubsection:PhiMaxOne}

\noindent In this subsection, we compute the  bottom free space boundary of the first
cell (over $[0, 1]\times [0,1]$).
The optimization problem~\eqref{opt:Phi}  
for $\Phi= \Phi_{\max}$ has the same solution as:
{ 
\begin{alignat*}{2} 
  \text{minimize  }\  &\  \lambda\\[-1mm]
  \text{such that }\  &
\begin{array}{l}
\forall \text{ points } \bp\in   \lambda\cdot R_1(0) + (1-\lambda) \cdot R_1(1),\\
\forall \text{ points } \bq\in  R_2(0)\\
\end{array}\\
& \qquad \qquad \text{ we have  } \norm{\bp - \bq} \leq \delta\\
 & \ 0 \leq  \lambda \leq 1
   \end{alignat*}
}
Unfortunately, this cannot be converted into an LP instance
as in the $\Phi_{\min}$ case  because of the
``for all'' quantifier in the constraints.
The above optimization problem can be expressed in the theory of reals which is 
decidable~\cite{Basu2006}.
This gives us a procedure to compute the free space cell boundaries for
$\Phi_{\max}$.
Once we have the free space boundaries, we can use a dynamic programming 
algorithm (as in the $\Phi_{\min}$ case) to obtain the following result.

\begin{proposition}[$\dvarfre^{\dagger}$ Decision Problem]
\label{proposition:DecisionVarOne}
Given  \pprs $R_1, R_2$ represented as $m_1$, $m_2$  polytopes respectively ,
and a $\delta \geq 0$, it is decidable to check
$\dvarfre^{\dagger}(R_1, R_2) \leq \delta$ for both $L_1^{\max}$ and  $L_{\infty}$  norms on $\reals^{d+1}$.
\qed
\end{proposition}

The check in Proposition~\ref{proposition:DecisionVarOne} uses the theory of reals and
has a high complexity.
We show in the next subsection that under certain assumptions on the \pprs, we can obtain
a polynomial time procedure.

\subsubsection{$\mspace{-14mu}\bm{\Phi_{\max}}\mspace{-2mu}$ Free Space: Polynomial Time Special Case}
\label{subsubsection:PhiMaxPoly}

In this subsection, we obtain a polynomial time algorithm for computing the
free space for $\bm{\Phi_{\max}}$, under mild conditions on the \pprs.

For a \emph{fixed} 
$\lambda$, we can check if 
\[\Phi_{\max}\Big( \lambda\cdot R_1(0) + (1-\lambda) \cdot R_1(1), R_2(0)\Big) \leq \delta.\]
This is done as follows.
Consider the optimization problem
{ 
\begin{equation}
\label{opt:Discretize}
\begin{alignedat}{2} 
  \text{maximize  }\  &\  \Delta\\[-1mm]
  \text{such that }\  &
\norm{\lambda \cdot \bx^0 + (1-\lambda) \cdot \bx^1 - 
\by} \geq \Delta\\
& A_1^0\cdot \bx^0 \leq \bb_1^0\\
& A_1^1\cdot \bx^1 \leq \bb_1^1\\
& A_2\cdot \by \leq \bb_2\\
& 0 \leq  \Delta 
\end{alignedat}
\end{equation}
}
The following cases arise.
\begin{compactitem}
\item 
If the optimal $\Delta$ is strictly bigger than $\delta$, then 
\[
\Phi_{\max}\big( \lambda\cdot R_1(0) + (1-\lambda) \cdot R_1(1),\, R_2(0)\big) > \delta
\]
because in this case the constraints in~\eqref{opt:Discretize} imply that there exist
points $\bx^0\in R_1(0)$ and $ \bx^1\in R_1(1)$ and
 $\by \in R_2(0)$ such that $\norm{\lambda \cdot \bx^0 + (1-\lambda) \cdot \bx^1 - 
\by} \geq \Delta > \delta$.
Hence  $\tuple{\lambda, 0}$ does not belong to the free space.
\item 
If  $\Delta\! \leq \!\delta$, it implies that
 $\Phi_{\max}\big( \lambda\vdot R_1(0) + (1\!-\!\lambda) \vdot R_1(1),\,  R_2(0)\big)
\leq \delta$.
Hence  $\tuple{\lambda, 0}$ belongs to the free space.
\end{compactitem}
Finally, note that the feasible region of
\eqref{opt:Discretize} is never empty since for
$\Delta= 0$ the variables $\bx^0, \bx^1, \by$ can range over
values in $R_1(0), R_1(1), R_2(0)$ respectively;
hence one of the above cases will hold.
Problem~\eqref{opt:Discretize} can be framed as an LP instance by adding additional
variables using the same methods as in the case for $\Phi_{\min}$ for 
$L_1^{\max}$  or $L_{\infty}$ norms.

If we can find \emph{one} $\lambda$ value such that
$\Phi_{\max}\big( \lambda\cdot R_1(0) + (1-\lambda) \cdot R_1(1),\,  R_2(0)\big) \leq \delta$,
then we can do binary search over the interval $[0, \lambda]$
to get $\lambda^{\min}$ (and similarly for $\lambda^{\max}$).
We next present a heuristic to do this in polynomial time.
Fix an integer $K$,
 partition $[0,1]$ into $K$ equal intervals, and  check for
$\lambda = 0, \frac{1}{K}, \frac{2}{K}, \dots, 1$ whether 
$\tuple{\lambda, 0}$ belongs to
the free space.

Once the first $\lambda \in \set{0, \frac{1}{K}, \frac{2}{K}, \dots, 1}$ is found 
such that $\tuple{\lambda, 0}$ belongs to
the free space, we perform a binary search around it over the interval
$(\lambda-1/K,\, \lambda]$ to obtain $\lambda^{\min}$  to a desired degree of
accuracy (which we take to be less than $2^{-cK}$ for a constant $c$ for convenience),
and similarly for $\lambda^{\max}$.
If the binary search fails to obtain a lower or upper boundary, we set the corresponding
lower or upper boundary to $\lambda$.
In total, 
 we solve $O(K)$ instances of  problem~\eqref{opt:Discretize}.
Suppose that the actual free space interval at the  bottom boundary of the
cell is $[\lambda^{\min}, \lambda^{\max}] \times \set{0}$.
If $\lambda^{\max}-\lambda^{\min} < 1/K$, we \emph{may} find an empty
subinterval. 
If $\lambda^{\max}-\lambda^{\min} \geq  1/K$, we are \emph{guaranteed} to find
the interval (to any desired degree of accuracy).
%

Observe that if the bottom boundary of cell $i,j$ is $[\lambda^{\min},\lambda^{\max}]\times\set{j}$, then
it means that the set of \emph{all} optimal reparameterizations $\alpha_1, \alpha_2$  in Equation~\eqref{equation:PhiMaxReframe} in addition satisfy
$\big(\alpha_2(\theta) = j\big) \rightarrow \big(\alpha_1(\theta) \in [\lambda^{\min},\lambda^{\max}]\big)$.
In other words, the polytope at time $\alpha_2(\theta)$ in the \ppr $R_2$ 
can only be mapped to $R_1$ polytopes in between times  $[\lambda^{\min},\lambda^{\max}]$.
The smaller the interval $[\lambda^{\min},\lambda^{\max}]$,
the more restricted the 
allowable timing distortions which
witness $\dvarfre^{\dagger}\!\left(R_1, R_2\right) \leq \delta$, and thus, the smaller the degree of freedom of time-distorting of the time-point $j$ in $R_2$; which in turn means the less
robust the possible reparameterizations..

\begin{proposition}[$\Phi_{\max}$ Free Space in Polyomial time]
\label{proposition:FreeSpaceBoundaryMaxPoly}
Given two \pprs $R_1$ and $R_2$,  the set $\free_{\delta}^{\Phi_{\max}}(R_1,R_2)$ at the
boundaries of cell $i,k$ can be computed to a precision of $O(K)$ bits in time
$O\left(K\cdot \LP\left(S_1^i+S_1^{i+1}+S_2^k+S_2^{k+1}\right)\right)$, 
provided the free space intervals at the cell boundaries, if non-empty, 
are of length at least $\frac{1}{K}$,
where
$S_j^l$ denotes the halfspace representation size of polytope
$R_j(l)$, and $\LP(\cdot)$ is the (polynomial time) upper bound for solving
linear programming.\qed
\end{proposition}
This gives us the following decision procedure using a dynamic programming algorithm,
and  improves 
Proposition~\ref{proposition:DecisionVarOne}
time complexity if the \pprs satisfy certain
conditions.

\begin{proposition}[$\dvarfre^{\dagger}$ Decision Problem in Polynomial Time]
\label{proposition:DecisionVarPoly}
Given   \pprs $R_1, R_2$ represented by $m_1$, $m_2$  polytopes respectively,
$\delta\! \geq\! 0$, and  integer $K\!> \!0$, 
we can  
decide the question
$\dvarfre^{\dagger}(R_1, R_2) \leq \delta$ under the two conditions:
\begin{compactenum}
\item 
 $\forall\ i\in \set{0.. m_1}$, and $\forall\ j \in \set{0..m_2-1}$, either
(a)~
there exists a sub-interval $[\lambda^{\min}, \lambda^{\max}]  \subseteq [j, j+1]$,
with $\lambda^{\max}- \lambda^{\min} \geq 1/K$,  such that
$\Phi_{\max}\left(R_1(i), R_2(t)\right)\leq \delta $ for all
$t\in [\lambda^{\min}, \lambda^{\max}] $, or
(b)~for all $t\in [j, j+1]$, we have $\Phi_{\max}\left(R_1(i), R_2(t)\right)> \delta $; and
\item 
 $\forall\ j \in  \set{0..m_2}$, and  $\forall\ i \in\set{0.. m_1-1}$, either
(a)~there exists a sub-interval $[\lambda^{\min}, \lambda^{\max}]  \subseteq [i, i+1]$,
with $\lambda^{\max}- \lambda^{\min} \geq 1/K$,
 such that
$\Phi_{\max}\left(R_1(t), R_2(j)\right)\leq \delta $ for all
$t\in [\lambda^{\min}, \lambda^{\max}] $, or
(b)~for all $t\in [i, i+1]$, we have $\Phi_{\max}\left(R_1(t), R_2(j)\right)> \delta $
\end{compactenum}
in time 
$O\left(m_1\vdot m_2\vdot K\vdot \LP(S_{\max})\right)$
for both $L_1^{\max}$,  $L_{\infty}$ norms
where
$S_{\max}$ is the maximum of the halfspace representation sizes of the given
polytopes, and $\LP()$ is the (polynomial time) upper bound for solving
linear programming.
\qed
\end{proposition}


%
An analysis of the  dynamic programming reachability algorithm shows that
the two conditions in  Proposition~\ref{proposition:DecisionVarPoly} are
only required for an $i,j$ pair collection  for which  a cell-$i,j$ from the collection
occurs in \emph{every} path
from $0,0$ to $m_1, m_2$ in the free space diagram of the
two \pprs.
As a result, for a sufficiently large $K$, we expect the algorithm of this subsection
to work in all except for certain pathological cases.

Proposition~\ref{proposition:DecisionVarPoly} gives us a conservative
procedure in case the validity of the two stated conditions is not known:
if for a chosen $K >0$, the procedure returns that the distance is less than or equal to $\delta$,
then indeed $\dvarfre^{\dagger}(R_1, R_2) \leq \delta$.
Also note that as $\delta$ increases, the corresponding free space
and the free space boundaries become larger, and when $\delta$ is increases enough,
 the \ppr conditions are satisfied.
Since we intend to use the $\dvarfre^{\dagger}$ distances
of \pprs  as over-approximations of tracepipes, the conservative
nature of Proposition~\ref{proposition:DecisionVarPoly} does not 
break the over-approximation scheme.

\section{Variation  Distance Bounds}
\label{section:FinalAlgo}

We now put everything together,
using the results of the preceding sections to obtain
bounds on the variation distance $\dvarskoro(F_1, F_2)$ for
\pprs $F_1$ and $F_2$.
From 
Propositions~\ref{proposition:OverUnder},~\ref{proposition:SkoroToFrechet},
and Theorems~\ref{theorem:ReachpipeFreVar},~\ref{theorem:ReachpipeFreMin},
and using binary search on the decision algorithms of 
Propositions~\ref{proposition:DecisionMin} and~\ref{proposition:DecisionVarOne}
we get the following theorem.
\begin{theorem}
\label{theorem:Final}
Suppose tracepipes $F_1$ and $ F_2$ correspond to
sampled over-approximate  reach set polytopes $\ceil{\rpipe(F_1)}(t_1^1), \dots, 
\ceil{\rpipe(F_1)}(t_1^{m_1})$ at time-points
$t_1^1, \dots, t_1^{m_1}$, and $\ceil{\rpipe(F_2)}(t_2^1), \dots, 
\ceil{\rpipe(F_2)}(t_2^{m_2})$ at time-points
$t_2^1, \dots, t_2^{m_2}$ respectively.
Let $\ceil{\rpipe(F_1)}$ and $\ceil{\rpipe(F_1)}$ be
corresponding reachpipe completions constructed by  linear interpolation.
We can compute $\beta_{\min},  \beta_{\max}$ with 
\[
\beta_{\min} \leq   \dvarskoro(F_1, F_2) \leq \beta_{\max}
\]
for the Skorokhod trace metric over $L_1, L_{\infty}$ norms on $\reals^d$ such that
\begin{compactitem}
\item $ \beta_{\min} = \dminskoro\!\left(\fpipe\big(\ceil{\rpipe(F_1)}\big), \,
\fpipe\big(\ceil{\rpipe(F_2)}\big)\right)$
and 
\item $  \beta_{\max}$ is an upper-bound of the variation distance
$ \dvarskoro\!\left(\fpipe\big(\ceil{\rpipe(F_1)}\big), \,
\fpipe\big(\ceil{\rpipe(F_2)}\big)\right)$; and is  equal to the
the Skorokhod distance $\dvarskoro^{\dagger}(\ceil{\rpipe(F_1)}, \ceil{\rpipe(F_2)})$
between the reachpipes 
$\ceil{\rpipe(F_1)}$ and $\ceil{\rpipe(F_2)}$
(where $\dvarskoro^{\dagger}$ is defined analogously to $\dvarfre^{\dagger}$).\qed
\end{compactitem}
\end{theorem}
In order to do binary searches on the decision procedures used in 
Theorem~\ref{theorem:Final}, we need an upper bound $U$ on $\beta_{\max}$.
This upper bound can be obtained as follows (in polynomial time).
We pick  one pair of reparameterizations  and use these to get an upper bound $U$
on $\dvarfre^{\dagger}(R_1, R_2)$ (and thus on  $\dvarskoro^{\dagger}(R_1, R_2)$)
for $R_1= \ceil{\rpipe(F_1)}$, and $R_2= \ceil{\rpipe(F_2)}$.
Assume $m_2\geq m_1$.
Fix $\alpha_1: [0, 1]\rightarrow [0, m_1]$ to be any non-decreasing reparameterization
such that $\alpha_1(\theta) = m_1$ for $\theta \geq 0.5$;
and let  $\alpha_2: [0, 1]\rightarrow [0, m_1]$ be a non-decreasing reparameterization
such that $\alpha_2(\theta) = \alpha_1(\theta)$ for $\theta \leq 0.5$, and
$\alpha_2$ over $[0.5, 1]$ being non-decreasing to $[m_1, m_2]$.
An  upper bound of $\dvarfre^{\dagger}(R_1, R_2)$ is
\begin{equation}
\label{equation:UpperBoundIdentity}
\max_{0\leq \theta \leq 1} \Phi_{\max}\big(
R_1\!\left(\alpha_1(\!\theta\!)\right),\,
R_2\!\left(\alpha_2(\!\theta\!)\right)
\!\!\big)
\end{equation}
The stated reparameterizations are such that $R_1(i)$ is compared to
$R_2(i)$ for $0\leq i\leq m_1$ in $\Phi_{\max}$, and
$R_2(i)$ for $i> m_1$ is compared to $ R_1(m_1)$.
It can be shown that the value of Expression~\eqref{equation:UpperBoundIdentity}
is the maximum of
$ \max_{i\in \set{0, 1,\dots, m_1}}\Phi_{\max}\left(R_1(i), R_2(i)\right)$
and
$ \max_{j \in  \set{m_1,\dots, m_2}}\Phi_{\max}\left(R_1(m_1), R_2(j)\right)$.
These two maximums can be computed in polynomial time by computing $\Phi_{\max}\left(R_1(i), R_2(j)\right)$
for required $i,j$ pairs using linear
programming
(Lemmas~\ref{lemma:CoarseUpperBound}, and~\ref{lemma:PhiMaxCompute}
in the Appendix).
Once the upper bound $U$ is obtained, we can compute
$\beta_{\min}$ in $O\left(\left(\lg(U) + B\right)\cdot
 m_1\cdot m_2  \cdot\LP(S_{\max})\right)$ time, where
$B$ is the number of desired bits of the fractional part in $ \beta_{\min}$, and
$S_{\max}$ is the maximum of the halfspace representation sizes of the given
polytopes, and $\LP(\cdot)$ is the (polynomial time) upper bound for solving
linear programming.

\smallskip\noindent\textbf{Polynomial Time Case for $\beta_{\max}$.}
 Theorem~\ref{theorem:Final} uses the theory of reals to
obtain $\beta_{\max}$.
In case an upper bound $U$ on $\beta_{\max}$ is given and the \pprs and $\delta< U$ are such that the conditions of
Proposition~\ref{proposition:DecisionVarPoly} are satisfied,  we can employ the 
polynomial time algorithm of the proposition in the decision question
queries for  obtaining $\beta_{\max}$.
This procedure runs in  $O\left(\left(\lg(U) + B\right)\cdot
 m_1\cdot m_2 \cdot K \cdot\LP(S_{\max})\right)$ time, where
 $K$ is an integer governing the robustness of retiming functions
 (in the sense discussed above Proposition~\ref{proposition:FreeSpaceBoundaryMaxPoly}).
Note that if the  \pprs do not satisfy the  the conditions of
Proposition~\ref{proposition:DecisionVarPoly}, then this procedure will still give an upper
bound on $ \dvar\!\left(\fpipe\big(\ceil{\rpipe(F_1)}\big), \,
\fpipe\big(\ceil{\rpipe(F_2)}\big)\right)$,
but it may be larger than the Skorokhod distance $\dvarskoro^{\dagger}(\ceil{\rpipe(F_1)}, \ceil{\rpipe(F_2)})$ between the reachpipes 
$\ceil{\rpipe(F_1)}$ and $\ceil{\rpipe(F_2)}$.

%




\smallskip\noindent\textbf{Using Sliding Windows.}
The Skorokhod metric 
allows matching an $F_1$ trace segment
in between times $t_1^0, t_1^1$ to  $F_2$ trace segments in between
times  $t_2^{m_2-1}, t_2^{m_2}$, \ie, the retimings put no limit on the timing distortions.
In practice, we have bounds on timing distortions.
As a result, we can  restrict the retimings to be in
a window $W$:
we require that trace segment $j$ of one trace only be matched
to portions of other traces consisting of segments
$j\!-\!W$ though $j\!+\!W$.
Under this restriction, the algorithm of Theorem~\ref{theorem:Final} can
be improved to run in time $O\left( (\left(\lg(U) + B\right)\cdot m\cdot W\cdot  K \cdot\LP(S_{\max})\right)$,
where $m\!=\!\max(m_1, m_2)$.
Usually $W, B$ and $K$ can be taken to be constants, 
thus we get a practical running time of
$O\left( m  \cdot \lg(U)\cdot\LP(S_{\max})\right)$, which is linear in the number of
given polytope reachsets, and linear in the LP solving time involving the largest
given polytope representation.

\section{Conclusions}

We have considered the problem of determining the distance between two tracepipes.
Such  problems arise in the analysis of dynamical systems under the presence of uncertainties
and noise.
Our starting point was the polynomial-time algorithm 
to compute the Skorokhod metric between individual traces \cite{MajumdarP15}.
Our algorithm takes as input discrete sequences of polyhedral approximations to the reach set, such as those provided
by symbolic tools such as SpaceEx \cite{GirardGM06,FrehseGDCRLRGDM11}. 
Our main result shows polynomial time algorithms to approximate the distance from above and from below.

\smallskip\noindent\textbf{Acknowledgements.}
The authors thank Fernando Pereira for helpful discussions; and
Raimund Seidel for pointing out the interpretation of
reachpipes as set-valued traces for applying the free-space technique.

{
\bibliographystyle{plain}
\bibliography{skoro-pipe}
}
\renewcommand{\baselinestretch}{1.0}

\section{Appendix}

\begin{proof}[\textbf{Proof of Theorem~\ref{theorem:ReachpipeFreMin}}]

We prove inequalities in both directions.

\medskip\noindent
(1)~$\dminfre\left(\fpipe(R_1), \fpipe(R_2)\right) \geq \dminfre^{\dagger}\!\left(R_1, R_2\right) $.\\
Consider any $f_1\in  \fpipe(R_1)$, and any  $f_2\in  \fpipe(R_2)$.
We have 
\[
\mspace{-20mu}
\dist_{\fre}(f_1, f_2) =  \inf_{\substack{
\alpha_1: [0,1] \rightarrow [0, m_1] \\
\alpha_2: [0,1] \rightarrow [0, m_2] }}\ 
\max_{0\leq \theta \leq 1} \norm {f_1\left(\alpha_1(\theta)\right) - f_2\left(\alpha_2(\theta)\right)}
\]
As in the proof of Theorem~\ref{theorem:ReachpipeFreVar},
we have that 
for every $\alpha_1, \alpha_2, \theta$, 
\[
\norm {f_1\left(\alpha_1(\theta)\right) - f_2\left(\alpha_2(\theta)\right)}
\geq 
 \Phi_{\min}\left(
R_1\left(\alpha_1(\theta)\right), \, R_2\left(\alpha_2(\theta)\right)
\right)
\]
Thus, for every $f_1\in  \fpipe(R_1)$, and  $f_2\in  \fpipe(R_2)$, we have
\[
\dist_{\fre}(f_1, f_2) \geq  \inf_{\substack{
\alpha_1: [0,1] \rightarrow [0, m_1] \\
\alpha_2: [0,1] \rightarrow [0, m_2] }}\ 
\max_{0\leq \theta \leq 1}
\Phi_{\min}\Big(
R_1\!\left(\alpha_1(\!\theta\!)\right),\,
R_2\!\left(\alpha_2(\!\theta\!)\right)
\Big)
\]
i.e., $\dist_{\fre}(f_1, f_2)  \geq \dminfre^{\dagger}\!\left(R_1, R_2\right)$.
This implies that 
$\inf_{f_1\in \fpipe(R_1), f_2\in \fpipe(R_2)} \dist_{\fre}(f_1, f_2)
\geq \dminfre^{\dagger}\!\left(R_1, R_2\right)$.
This completes the proof of the first direction.

\medskip\noindent
(2)~$\dminfre\left(\fpipe(R_1), \fpipe(R_2)\right) \leq \dminfre^{\dagger}\!\left(R_1, R_2\right) $.\\
Recall that $\dminfre\left(\fpipe(R_1), \fpipe(R_2)\right)  =$
\[
 \inf_{f_1\in \fpipe(R_1), f_2\in \fpipe(R_2)}
 \inf_{\substack{
\alpha_1: [0,1] \rightarrow [0, m_1] \\
\alpha_2: [0,1] \rightarrow [0, m_2] }}\ 
\max_{0\leq \theta \leq 1} \norm {f_1\left(\alpha_1(\theta)\right) - f_2\left(\alpha_2(\theta)\right)}.
\]
This equals (switching the $\inf$ order):
\[
 \inf_{\substack{
\alpha_1: [0,1] \rightarrow [0, m_1] \\
\alpha_2: [0,1] \rightarrow [0, m_2] }}\ 
 \inf_{f_1\in \fpipe(R_1), f_2\in \fpipe(R_2)}\ 
\max_{0\leq \theta \leq 1} \norm {f_1\left(\alpha_1(\theta)\right) - f_2\left(\alpha_2(\theta)\right)}.
\]
We need to show that the above expression is $\leq$ than:
\[
\inf_{\substack{
\alpha_1: [0,1] \rightarrow [0, m_1] \\
\alpha_2: [0,1] \rightarrow [0, m_2] }}\ 
\max_{0\leq \theta \leq 1}\  \Phi_{\min}\Big(
R_1\!\left(\alpha_1(\!\theta\!)\right),\,
R_2\!\left(\alpha_2(\!\theta\!)\right)
\Big)
\]
To show this direction of the inequality, it suffices to show that
for every pair of valid reparameterizations $\alpha_1, \alpha_2$, we have:
\begin{multline}
\label{equation:MinInqTheoremProof}
\inf_{f_1\in \fpipe(R_1), f_2\in \fpipe(R_2)}\ 
\max_{0\leq \theta \leq 1} \norm {f_1\left(\alpha_1(\theta)\right) - f_2\left(\alpha_2(\theta)\right)}\\[-1.5mm]
\leq\\[-1mm]
\max_{0\leq \theta \leq 1}\  \Phi_{\min}\Big(
R_1\!\left(\alpha_1(\!\theta\!)\right),\,
R_2\!\left(\alpha_2(\!\theta\!)\right)
\Big)
\end{multline}
The formal proof of the above inequality is technical.
%
%
We sketch the main ideas.
Fix $\alpha_1, \alpha_2$ reparameterizations.
Define the function $\minpairs$  from $ [0,1] $ to
subsets of $\reals^{d+1} \times \reals^{d+1}$ as
 $\minpairs(\theta) =$
\[
\left\{
\tuple{\bp_1, \bp_2} \ \left\arrowvert\ 
\begin{array}{l}
\bp_1\in R_1\!\left(\alpha_1(\theta)\right),
\text{ and } 
\bp_2\in R_2\!\left(\alpha_2(\theta)\right),\\
\text{ and }
\norm{\bp_1-\bp_2} \leq \Phi_{\min}
\Big(
R_1\!\left(\alpha_1(\!\theta\!)\right),\,
R_2\!\left(\alpha_2(\!\theta\!)\right)
\Big)
\end{array}
\right.\mspace{-10mu}\right\}
\]  
That is, $\minpairs(\theta)$ contains point pairs
$\tuple{\bp_1, \bp_2}$ with $\bp_1\in R_1\!\left(\alpha_1(\theta)\right)$,
and 
$\bp_2\in R_2\!\left(\alpha_2(\theta)\right)$ such that
$\bp_1, \bp_2$ are the closest points in the corresponding polytopes
$ R_1\!\left(\alpha_1(\theta)\right)$ and 
$R_2\!\left(\alpha_2(\theta)\right)$
(there may be several such pairs for the two polytopes).
It can be shown that for each $\theta$, we can pick a single point tuple
from $\minpairs(\theta) $, namely
$\tuple{\bp_1^{\theta}, \bp_2^{\theta}}$ such that
the functions $\curve_1\left(\alpha_1(\theta)\right) = \bp_1^{\theta}$ and
$\curve_2\left(\alpha_2(\theta)\right) = \bp_2^{\theta}$ are continuous functions from 
$[0, m_1]$ and $[0, m_2]$ to $\reals^{d+1}$, ie. they are continuous traces.
This can be done due to the fact that $R_1$ and $R_2$ are \pprs and thus the
polygons $ R_1\!\left(\alpha_1(\theta)\right)$ and 
$R_2\!\left(\alpha_2(\theta)\right)$ change smoothly with respect to $\theta$.

Observe that the curves $\curve_1$ and $\curve_2$ are such that
\begin{multline*}
\max_{0\leq \theta \leq 1} \norm {\curve_1\!\left(\alpha_1(\theta)\right) - \curve_2\!\left(\alpha_2(\theta)\right)}\\[-1.5mm]
\leq\\[-1mm]
\max_{0\leq \theta \leq 1}\  \Phi_{\min}\Big(
R_1\!\left(\alpha_1(\!\theta\!)\right),\,
R_2\!\left(\alpha_2(\!\theta\!)\right)
\Big)
\end{multline*}
This prove Inequality~\ref{equation:MinInqTheoremProof}.
This concludes the second part of the  theorem proof.
\qedhere
\end{proof}

\medskip
\begin{proof}[\textbf{Proof of Lemma~\ref{lemma:MinNew}}]

The basic idea is that we introduce variables
$\bz^0= \lambda\cdot \bx^0$ and 
$\bz^1= (1-\lambda)\cdot \bx^0$, and we multiply both sides of
 $A_1^0\cdot \bx^0 \leq \bb_1^0$ by $\lambda$, and
of $A_1^1\cdot \bx^1 \leq \bb_1^1$ by $1-\lambda$.
For the two optimization problems to be the same, it suffices to show that
for any $0\leq \lambda\leq 1$, and for
any $\by$ satisfying $A_2\cdot \by \leq \bb_2$, 
{
\begin{equation}
\label{opt:XDomain}
\begin{alignedat}{2} 
 \text{there exist } \bx^0, \bx^1
\text{  such that: }\quad & \norm{\lambda \cdot \bx^0 + (1-\lambda) \cdot \bx^1 - 
\by} \leq \delta\\
  \text{with }\  &
 A_1^0\cdot \bx^0 \leq \bb_1^0\\
& A_1^1\cdot \bx^1 \leq \bb_1^1
\end{alignedat}
\end{equation}
}
iff 
{
\begin{equation}
\label{opt:ZDomain}
\begin{alignedat}{2} 
 \text{there exist } \bz^0, \bz^1
\text{  such that: }\quad &  \norm{\bz^0 + \bz^1  - 
\by} \leq \delta\\ 
  \text{with  }\  &
 A_1^0\cdot \bz^0 \leq \lambda\cdot \bb_1^0\\
& A_1^1\cdot \bz^1 \leq (1-\lambda)\cdot \bb_1^1
\end{alignedat}
\end{equation}
}
Fix a $\lambda$, and a $\by$ vector.
We show the above equivalence.

\noindent
\textbf{``Only if''.}
Suppose there exist $\bx^0 , \bx^1$ satisfying constraints~\ref{opt:XDomain}.
Let $\bz^0= \lambda\cdot \bx^0$ and 
$\bz^1= (1-\lambda)\cdot \bx^0$.
Observe that  $ \bz^0, \bz^1$ satisfy the conditions of the second system, and also
$\norm{\bz^0 + \bz^1  - 
\by} \leq \delta$ as $\norm{\lambda \cdot \bx^0 + (1-\lambda) \cdot \bx^1 - 
\by} \leq \delta$.
This concludes the proof of the ``Only if'' direction.

\noindent
\textbf{``If''.}
Suppose there exist $\bz^0 , \bz^1$ satisfying constraints~\ref{opt:ZDomain}.
If $\lambda \neq 0$ and $\lambda \neq 1$, then take 
$\bx^0 = \frac{1}{\lambda}\cdot \bz^0$, and 
$\bx^1 = \frac{1}{1-\lambda}\cdot \bz^1$.
It can be checked that $\bx^0, \bx^1$ satisfy constraints~\ref{opt:XDomain}.

Now suppose $\lambda = 0$.
The point $\bz^0 , \bz^1$ thus also satisfy:
{
\begin{equation*}
\begin{alignedat}{2} 
  \norm{\bz^0 + \bz^1  - 
\by} \leq \delta\\ 
  \qquad \qquad  \text{ with  }\ &
 A_1^0\cdot \bz^0 \leq \bm{0}\\
& A_1^1\cdot \bz^1 \leq  \bb_1^1
\end{alignedat}
\end{equation*}
}
If $\bb^0=\bm{0}$, then $\bx^0 = \bz^0$ and $\bx^1 = \bz^1$
satisfy constraints~\ref{opt:XDomain}.

Suppose $\bb^0\neq \bm{0}$.
From Lemma~\ref{lemma:PolytopeZero}, since $A_1^0\cdot \bz^0 \leq \bm{0}$,
we must have that $\bz^0 = \bm{0}$.
Thus,  we have $\norm{ \bz^1  - 
\by} \leq \delta$ with $A_1^1\cdot \bz^1 \leq  \bb_1^1$.
Now we let $\bx^0$ be any point in the polytope 
$A_1^0\cdot \bx^0 \leq \bb_1^0$, and $\bx^1 = \bz^1$.
It can be seen that these $ \bx^0, \bx^1$ satisfy
{
\begin{equation*}
\begin{alignedat}{2} 
   \norm{ \bx^1  - 
\by} \leq \delta\\ 
 \qquad\qquad   \text{ with  }\  &
 A_1^0\cdot \bx^0 \leq \bb^0\\
& A_1^1\cdot \bx^1 \leq  \bb_1^1
\end{alignedat}
\end{equation*}
}
The case of $\lambda = 1$ is similar.
This concludes the proof of he ``If'' part, and thus also the proof of the lemma.
\end{proof}

\begin{lemma}
\label{lemma:CoarseUpperBound}
Let $R_1, R_2$ be \pprs represented by $m_1, m_2$ polytopes respectively with
$m_2\geq m_1$.
Fix $\alpha_1: [0, 1]\rightarrow [0, m_1]$ to be any non-decreasing reparameterization
such that $\alpha_1(\theta) = m_1$ for $\theta \geq 0.5$;
and let  $\alpha_2: [0, 1]\rightarrow [0, m_1]$ be a non-decreasing reparameterization
such that $\alpha_2(\theta) = \alpha_1(\theta)$ for $\theta \leq 0.5$, and
$\alpha_2$ over $[0.5, 1]$ being non-decreasing to $[m_1, m_2]$.
The value of   $\dvarfre^{\dagger}(R_1, R_2)$ is at most
the maximum of
$ \max_{i\in \set{0, 1,\dots, m_1}}\Phi_{\max}\left(R_1(i), R_2(i)\right)$
and
$ \max_{j \in  \set{m_1,\dots, m_2}}\Phi_{\max}\left(R_1(m_1), R_2(j)\right)$.
\end{lemma}
\begin{proof}
Since $\alpha_1, \alpha_2$ are valid non-decreasing reparameterizations,
we  have
\[
\dvarfre^{\dagger}(R_1, R_2) \ \leq
\ 
\max_{0\leq \theta \leq 1} \Phi_{\max}\big(
R_1\!\left(\alpha_1(\!\theta\!)\right),\,
R_2\!\left(\alpha_2(\!\theta\!)\right)
\!\!\big)
\]
It is clear that
$\max_{0\leq \theta \leq 1} \Phi_{\max}\big(
R_1\!\left(\alpha_1(\!\theta\!)\right),\,
R_2\!\left(\alpha_2(\!\theta\!)\right)
\!\!\big)$ cannot be smaller than the
maximum of $ \max_{i\in \set{0, 1,\dots, m_1}}\Phi_{\max}\left(R_1(i), R_2(i)\right)$
and
$ \max_{j \in  \set{m_1,\dots, m_2}}\Phi_{\max}\left(R_1(m_1), R_2(j)\right)$.
We prove that the  two quantities are equal.
To prove this, it suffices to show that
if 
(a)~$\Phi_{\max}\left(R_1(i), R_2(i)\right) \leq \delta$, and
(b)~$\Phi_{\max}\left(R_1(i+1), R_2(i+1)\right) \leq \delta$,
then
for all $0\leq\lambda \leq 1$,
we have 
\[
\Phi_{\max}\left(
\begin{array}{l}
\lambda\cdot R_1(i) + (1-\lambda)\cdot R_1(i+1),\\
\lambda\cdot R_2(i) + (1-\lambda)\cdot R_2(i+1)
\end{array}
\right) 
\leq \delta.
\]
We prove the above as follows.
Assume (a)~$\Phi_{\max}\left(R_1(i), R_2(i)\right) \leq \delta$, and
(b)~$\Phi_{\max}\left(R_1(i+1), R_2(i+1)\right) \leq \delta$.
Let $\bp_1^{i}\in R_1(i)$, and $\bp_1^{i+1}\in R_1(i+1)$, and
$\bp_2^{i}\in R_2(i)$, and $\bp_2^{i+1}\in R_2(i+1)$.
We have
\begin{multline*}
\norm{\lambda \bp_1^{i} + (1-\lambda) \bp_1^{i+1}
-
\left(\lambda \bp_2^{i} + (1-\lambda) \bp_2^{i+1}\right)
}\\
\leq  \norm{\lambda \left(\bp_1^{i} -  \bp_2^{i}\right)}
+ \norm{(1-\lambda) \left(\bp_1^{i+1} -  \bp_2^{i+1}\right)}\\
\leq \lambda \delta + (1-\lambda) \delta = \delta
\end{multline*}
This concludes the proof.
\end{proof}

\begin{lemma}
\label{lemma:PhiMaxCompute}
Let $Q_1$ and $Q_2$ be  polytopes in  $\ptopereals$.
The value $\Phi_{\max}(Q_1, Q_2) $ can be computed in
time $O\left( \LP(|Q_1| + |Q_2|)\right)$ where 
$|Q_1|$ and $|Q_2|$ denote the halfspace representation sizes
of the respective polytopes, 
and
 $\LP()$ is the (polynomial time) upper bound for solving
linear programming.
\end{lemma}
\begin{proof}
Let $Q_1$ have the halfspace representation $A_1^0\cdot \bx^1 \leq \bb_1$,
and let $Q_2$ be $A_2\cdot \bx^2 \leq \bb_2$ for
$\bx^1$ and $\bx^2$ column vectors  of $d+1$ variables taking values in $\reals$.
The value of $\Phi_{\max}(Q_1, Q_2) $ is the solution to the
following constraint problem:
{ 
\begin{equation}
\label{opt:PhiMaxCompute}
\begin{alignedat}{2} 
  \text{maximize  }\  &\  \Delta\\ 
  \text{such that }\  &
\norm{\bx^1- \bx^2} \geq \Delta\\
& A_1^0\cdot \bx^1 \leq \bb_1\\
& A_2\cdot \bx^2 \leq \bb_2\\
& 0 \leq  \Delta 
\end{alignedat}
\end{equation}
}
The optimization problem~\eqref{opt:PhiMaxCompute} can be solved using
linear programming.
Suppose the solution of the  optimization problem~\eqref{opt:PhiMaxCompute}
is $\delta$.
It means that
(i)~there are no points $\bp^1\in  Q_1$, and  $\bp^2\in  Q_2$ such that
$\norm{\bp^1- \bp^2} > \delta$, and
(ii)~there exist points $\bp^1\in  Q_1$, and  $\bp^2\in  Q_2$ 
such that
$\norm{\bp^1- \bp^2} \leq \delta$.
These two facts imply $\Phi_{\max}(Q_1, Q_2) = \delta$.
\end{proof}



\end{document}